\documentclass{article}

\usepackage{ijcai11}
\usepackage{times}

\usepackage{amsmath,amssymb,amsthm}
\usepackage{epsfig}
\usepackage{color}
\usepackage{bbm}
\usepackage{boxedminipage}
\usepackage{pifont}
\usepackage{xspace}
\usepackage{graphics}
\usepackage{misc}
\usepackage{url}

\title{Foundations for Uniform Interpolation and Forgetting \\in Expressive Description Logics}
\author{
  \begin{tabular}[c]{ccc}
Carsten Lutz &$\qquad$& Frank Wolter\\
\textnormal{Fachbereich Informatik} && \textnormal{Department of Computer Science}\\
\textnormal{Universit\"at Bremen, Germany} && \textnormal{University of Liverpool, UK}
  \end{tabular}
}

\begin{document}

\date{}
\maketitle

\begin{abstract}
  We study uniform interpolation and forgetting in the description
  logic \ALC.  Our main results are model-theoretic characterizations
  of uniform interpolants and their existence in terms of
  bisimulations, tight complexity bounds for deciding the existence of
  uniform interpolants, an approach to computing interpolants when
  they exist, and tight bounds on their size. We use a mix of
  model-theoretic and automata-theoretic methods that, as a
  by-product, also provides characterizations of and decision
  procedures for conservative extensions.
\end{abstract}

\section{Introduction}

In Description Logic (DL), a \emph{TBox} or \emph{ontology} is a
logical theory that describes the conceptual knowledge of an
application domain using a set of appropriate predicate symbols.  For
example, in the domain of universities and students, the predicate
symbols could include the concept names $\mn{Uni}$, $\mn{Undergrad}$,
and $\mn{Grad}$, and the role name $\mn{has}$\_$\mn{student}$. When
working with an ontology, it is often useful to eliminate some of the
used predicates while retaining the meaning of all remaining ones. For
example, when \emph{re-using} an existing ontology in a new
application, then typically only a very small fraction of the
predicates is of interest. Instead of re-using the whole ontology, one
can thus use the potentially much smaller ontology that results from
an elimination of the non-relevant predicates. Another reason for
eliminating predicates is \emph{predicate hiding}, i.e., an ontology
is to be published, but some part of it should be concealed from the
public because it is confidential
\cite{DBLP:conf/kr/GrauM10}. Finally, one can view the result of
predicate elimination as an approach to \emph{ontology summary}: the
resulting, smaller and more focussed ontology summarizes what the
original ontology says about the remaining predicates.

The idea of eliminating predicates has been studied in AI under the
name of \emph{forgetting a signature (set of predicates)~$\Sigma$},
i.e., rewriting a knowledge base $K$ such that it does not use
predicates from $\Sigma$ anymore and still has the same logical
consequences that do not refer to predicates from $\Sigma$ \cite{Lin}.
In propositional logic, forgetting is also known as \emph{variable
  elimination}~\cite{vindependence}. In mathematical logic, forgetting
has been investigated under the dual notion of \emph{uniform
  interpolation w.r.t.\ a signature $\Sigma$}, i.e., rewriting a
formula $\vp$ such that it uses \emph{only} predicates from $\Sigma$
and has the same logical consequences formulated \emph{only} in
$\Sigma$. The result of this rewriting is then the \emph{uniform
  interpolant} of $\vp$ w.r.t.\ $\Sigma$. This notion can be seen as a
generalization of the more widely known Craig interpolation.

Due to the various applications briefly discussed above, forgetting
und uniform interpolation receive increased interest also in a DL
context
\cite{DBLP:conf/webi/EiterISTW06,DBLP:conf/ecai/WangWTZ10,DBLP:conf/ausai/WangWTPA09,DBLP:conf/semweb/WangWTPA09,DBLP:journals/ai/KontchakovWZ10,DBLP:conf/ijcai/KonevWW09}.
Here, the knowledge base $K$ resp.\ formula $\vp$ is replaced with a
TBox~\Tmc.
%
%
In fact, uniform interpolation is rather well-understood in
lightweight DLs such as DL-Lite and $\mathcal{EL}$: there, uniform
interpolants of a TBox~\Tmc can often be expressed in the DL in which
\Tmc is formulated \cite{DBLP:journals/ai/KontchakovWZ10,DBLP:conf/ijcai/KonevWW09}
and practical experiments have confirmed the
usefulness and feasibility of their computation \cite{DBLP:conf/ijcai/KonevWW09}. The
situation is different for `expressive' DLs such as \ALC and its
various extensions, where much less is known. There is a thorough
understanding of uniform interpolation on the level of
\emph{concepts}, i.e., computing uniform interpolants of concepts
instead of TBoxes \cite{CCMV-KR06,DBLP:conf/ausai/WangWTPA09}, which
is also what the literature on uniform interpolants in modal logic is
about \cite{Visser,DBLP:conf/jelia/HerzigM08}. 
On the TBox level, a basic
observation is that there are very simple \ALC-TBoxes and signatures $\Sigma$ 
such that the uniform interpolant of \Tmc w.r.t.\ $\Sigma$ cannot be
expressed in $\ALC$ (nor in first-order predicate logic) \cite{KR06}. A scheme for approximating (existing or
non-existing) interpolants of \ALC-TBoxes was devised in
\cite{DBLP:conf/semweb/WangWTPA09}. In
\cite{DBLP:conf/ecai/WangWTZ10}, an attempt is made to improve this to
an algorithm that computes uniform interpolants of \ALC-TBoxes in an
exact way, and also decides their existence (resp.\ expressibility in
\ALC). Unfortunately, that algorithm turns out to be incorrect.

The aim of this paper is to lay foundations for uniform interpolation
in \ALC and other expressive DLs, with a focus on (i)~model-theoretic
characterizations of uniform interpolants and their existence;
(ii)~deciding the existence of uniform interpolants and computing them
in case they exist; and (iii)~analyzing the size of uniform
interpolants. Clearly, these are fundamental steps on the way towards
the computation and usage of uniform interpolation in practical
applications.  Regarding~(i), we establish an intimate connection
between uniform interpolants and the well-known notion of a
bisimulation and characterize the existence of interpolants in terms
of the existence of models with certain properties based on
bisimulations. For~(ii), our main result is that deciding the
existence of uniform interpolants is 2-\ExpTime-complete, and that
methods for computing uniform interpolants on the level of concepts
can be lifted to the TBox level. Finally, regarding (iii) we prove that
the size of uniform interpolants is at most triple exponential in the
size of the original TBox (upper bound), and that, in general, no
shorter interpolants can be found (lower bound). In particular, this
shows that the algorithm from \cite{DBLP:conf/ecai/WangWTZ10} is
flawed as it always yields uniform interpolants of at most double
exponential size. Our methods, which are a mix of model-theory and
automata-theory, also provide model-theoretic characterizations of
conservative extensions, which are closely related to uniform
interpolation~\cite{KR06}.  Moreover, we use our approach to reprove the
2-\ExpTime upper bound for deciding conservative extensions from
\cite{KR06}, in an alternative and argueably more transparent way.


Most proofs in this paper are deferred to the appendix.

\section{Getting Started}
\label{sect:prelim}

We introduce the description logic ${\cal ALC}$ and define uniform
interpolants and the dual notion of forgetting.  Let $\NC$ and $\NR$
be disjoint and countably infinite sets of \emph{concept} and
\emph{role names}.  \emph{\ALC concepts} are formed using the
syntax rule
$$
C,D \longrightarrow \top \mid A \mid \neg C \mid C \sqcap D \mid \exists r .
C
$$
where $A\in \NC$ and $r\in \NR$.
The concept constructors $\bot$, $\sqcup$, and $\forall r.C$ are
defined as abbreviations: $\bot$ stands for $\neg \top$, $C \sqcup D$
for $\neg(\neg C \sqcap \neg D)$ and $\forall r . C$ abbreviates
$\neg \exists r. \neg C$. 
A \emph{TBox} is a finite set of \emph{concept inclusions} $C
\sqsubseteq D$, where $C,D$ are \ALC-concepts.
We use $C \equiv D$ as abbreviation for the two inclusions
$C \sqsubseteq D$ and $D \sqsubseteq C$. 

%
%
The semantics of \ALC-concepts is given in terms of
\emph{interpretations} $\Imc=(\Delta^\Imc,\cdot^\Imc)$, where
$\Delta^\Imc$ is a non-empty set (the \emph{domain}) and $\cdot^\Imc$
is the \emph{interpretation function}, assigning to each $A\in \NC$ a
set $A^\Imc \subseteq \Delta^\Imc$, and to each $r\in \NR$ a relation
$r^\Imc \subseteq \Delta^\Imc \times \Delta^\Imc$. The interpretation
function is inductively extended to concepts as follows:
$$
\begin{array}{l}
  \top^\Imc := \Delta^\Imc \quad
  (\neg C)^\Imc := \Delta^\Imc \setminus C^\Imc \quad
  (C \sqcap D)^\Imc := C^\Imc \cap D^\Imc \\[1mm]
  (\exists r . C)^\Imc := \{ d \in \Delta^\Imc \mid \exists e. (d,e) \in
  r^\Imc \wedge e \in C^\Imc \}
\end{array}
$$
An interpretation \Imc \emph{satisfies} an inclusion $C \sqsubseteq
D$ if $C^\Imc \subseteq D^\Imc$, and \Imc is a \emph{model} of a TBox
\Tmc if it satisfies all inclusions in \Tmc. 
A concept $C$ is
\emph{subsumed} by a concept $D$ relative to a TBox \Tmc (written
$\Tmc \models C \sqsubseteq D$) if every model \Imc of $\Tmc$
satisfies the inclusion $C \sqsubseteq D$. We write $\Tmc \models \Tmc'$
to indicate that $\Tmc \models C \sqsubseteq D$ for all $C \sqsubseteq D \in
\Tmc'$.

\smallskip

A set $\Sigma\subseteq \NC \cup \NR$ of concept and role names is
called a \emph{signature}. The signature ${\sf sig}(C)$ of a concept
$C$ is the set of concept and role names occurring in $C$, and
likewise for the signature $\mn{sig}(C \sqsubseteq D)$ of an inclusion $C \sqsubseteq D$
and $\mn{sig}(\Tmc)$ of a TBox \Tmc. A
\emph{$\Sigma$-TBox} is a TBox with $\mn{sig}(\Tmc) \subseteq \Sigma$,
and likewise for \emph{$\Sigma$-inclusions} and \emph{$\Sigma$-concepts}. 
%
\smallskip

We now introduce the main notions studied in this paper: 
uniform interpolants and conservative extensions.
\begin{definition}
\label{def:unifint}
  Let $\mathcal{T},\Tmc'$ be TBoxes and $\Sigma$ a signature. 
  \Tmc and $\Tmc'$ are \emph{$\Sigma$-inseparable} if for all
  $\Sigma$-inclusions $C \sqsubseteq D$, we have $\Tmc \models C
  \sqsubseteq D$ iff $\Tmc' \models C \sqsubseteq D$. We call
  \begin{itemize}
  \item $\Tmc'$ a
  \emph{conservative extension} of $\Tmc$ if $\Tmc'\supseteq \Tmc$ and
   $\Tmc$ and $\Tmc'$ are $\Sigma$-inseparable for $\Sigma= \mn{sig}(\Tmc)$.
  \item $\Tmc$ a \emph{uniform $\Sigma$-interpolant} of $\Tmc'$
  if $\mn{sig}(\Tmc)\subseteq \Sigma \subseteq \mn{sig}(\Tmc')$ and $\Tmc$ and $\Tmc'$ are 
  $\Sigma$-inseparable.
\end{itemize}
\end{definition}
Note that uniform $\Sigma$-interpolants are unique up
to logical equivalence, if they exist. 



The notion of forgetting as investigated in
\cite{DBLP:conf/ecai/WangWTZ10} is dual to uniform interpolation: a
TBox $\Tmc'$ is \emph{the result of forgetting about a signature
  $\Sigma$ in a TBox $\Tmc$} if $\Tmc'$ is a uniform $\mn{sig}(\Tmc)\setminus
\Sigma$-interpolant of $\Tmc$.
\begin{example}
{\em Let $\Tmc$ consist of the inclusions \\[1mm]
(1) ${\sf Uni}\sqsubseteq \exists 
{\sf has\_st}.{\sf Undergrad} \sqcap \exists {\sf has\_st}.{\sf Grad}$ \\[1mm]
(2) ${\sf Uni}\sqcap {\sf Undergrad}\sqsubseteq \bot$
$\quad$
(3) ${\sf Uni} \sqcap {\sf Grad}\sqsubseteq \bot$ \\[1mm]
(4) ${\sf Undergrad}\sqcap {\sf Grad}\sqsubseteq \bot$.\\[1mm] 
Then the TBox that consists of (2) and
$${\sf Uni}\sqsubseteq \exists 
{\sf has\_st}.{\sf Undergrad} \sqcap \exists {\sf has\_st}.
(\neg {\sf Undergrad}\sqcap \neg {\sf Uni})
$$
is the result of forgetting ${\{\sf Grad}\}$. Additionally forgetting
${\sf Undergrad}$ yields the TBox $\{{\sf Uni}\sqsubseteq \exists {\sf
  has\_st}.\neg{\sf Uni}\}$.}
\end{example}

The following examples will be used to illustrate our characterizations.
Proofs are provided once we have developed
the appropriate tools.

\begin{example}\label{ex1}
{\em In the following, we always forget $\{B\}$.

\smallskip
\noindent
(i) Let $\Tmc_1=\{A \sqsubseteq \exists r.B\sqcap \exists r.\neg B\}$ and 
$\Sigma_{1} = \{A,r\}$. Then $\Tmc_1'=
\{A \sqsubseteq \exists r.\top\}$
is a uniform $\Sigma_1$-interpolant of $\Tmc_{1}$. 

\smallskip
\noindent
(ii) Let $\Tmc_{2}=\{A \equiv B \sqcap \exists r.B\}$ and $\Sigma_{2}=\{A,r\}$.
Then $\Tmc'_{2}=\{A \sqsubseteq \exists r.(A \sqcup \neg \exists r.A)\}$
is a uniform $\Sigma_{2}$-interpolant of~$\Tmc_2$.

\smallskip
\noindent
(iii) For $\Tmc_3=\{A \sqsubseteq B, B \sqsubseteq \exists r.B\}$ and 
$\Sigma_{3}=\{A,r\}$, there is no uniform $\Sigma_{3}$-interpolant of
$\Tmc_3$.

\smallskip
\noindent
(iv) For $\Tmc_4=
\{A \sqsubseteq \exists r.B$, $A_0 \sqsubseteq \exists r.(A_1 \sqcap B)$,
$E \equiv A_1 \sqcap B \sqcap \exists r.(A_2 \sqcap B)\}$
and $\Sigma_{4}=\{A,r,A_{0},A_{1},E\}$, there is no uniform $\Sigma_{4}$-interpolant of 
$\Tmc_4$. Note that $\Tmc_4$ is of a very simple form, namely an acyclic
$\mathcal{EL}$-TBox, see \cite{DBLP:conf/ijcai/KonevWW09}.
}
\end{example}

 
 
Bisimulations are a central tool for studying the expressive power of
\ALC, and play a crucial role also in our approach to uniform
interpolants.  We introduce them next.  
A \emph{pointed
  interpretation} is a pair $(\Imc,d)$ that consists of an
interpretation $\Imc$ and a $d\in \Delta^{\Imc}$. 
\begin{definition}
  Let $\Sigma$ be a finite signature and $(\Imc_{1},d_{1})$,
  $(\Imc_{2},d_{2})$ pointed interpretations. A relation $S \subseteq
  \Delta^{\Imc_{1}}\times \Delta^{\Imc_{2}}$ is a
  \emph{$\Sigma$-bisimulation} between $(\Imc_{1},d_{1})$ and
  $(\Imc_{2},d_{2})$ if $(d_{1},d_{2})\in S$ and for all $(d,d') \in
  S$ the following conditions are satisfied:
\begin{enumerate}

  \item $d \in A^{\Imc_{1}}$ iff $d' \in A^{\Imc_{2}}$, for all $A \in \Sigma\cap \NC$;

  \item if $(d,e) \in r^{\Imc_{1}}$, then there exists $e'\in \Delta^{\Imc_{2}}$ such
that $(d',e') \in r^{\Imc_{2}}$ and $(e,e')\in S$, for all $r \in \Sigma\cap \NR$;

  \item if $(d',e') \in r^{\Imc_{2}}$, then there exists $e\in \Delta^{\Imc_{1}}$ such
that $(d,e) \in r^{\Imc_{1}}$ and $(e,e')\in S$, for all $r \in \Sigma\cap \NC$.

\end{enumerate}
$(\Imc_{1},d_{1})$ and $(\Imc_{2},d_{2})$ are
\emph{$\Sigma$-bisimilar}, written $(\Imc_{1},d_{1}) \sim_\Sigma
(\Imc_{2},d_{2})$, if there exists a $\Sigma$-bisimulation between
them.
\end{definition}
%
%
%
We now state the main connection between bisimulations and \ALC,
well-known from modal logic~\cite{GorankoOtto}.  Say that
$(\Imc_{1},d_{1})$ and $(\Imc_{2},d_{2})$ are
\emph{$\ALC_{\Sigma}$-equivalent}, in symbols
$(\Imc_{1},d_{1})\equiv_{\Sigma}(\Imc_{2},d_{2})$, if for all
$\Sigma$-concepts $C$, $d_{1}\in C^{\Imc_{1}}$ iff $d_{2}\in
C^{\Imc_{2}}$. An
\emph{interpretation} $\mathcal{I}$ has \emph{finite outdegree} if $\{
d' \mid (d,d') \in \bigcup_{r \in \NR} r^\Imc \}$ is finite, for all
\mbox{$d\in \Delta^{\mathcal{I}}$}.
\begin{theorem}\label{theorem:bisimchar}
  For all pointed interpretations $(\Imc_{1},d_{1})$ and
  $(\Imc_{2},d_{2})$ and all finite signatures~$\Sigma$,
  $(\Imc_{1},d_{1}) \sim_\Sigma (\Imc_{2},d_{2})$ implies
  $(\Imc_{1},d_{1})\equiv_{\Sigma}(\Imc_{2},d_{2})$; the converse
  holds for all $\Imc_1,\Imc_2$ of finite outdegree.
\end{theorem}
Bisimulations enable a purely semantic characterization of
uniform interpolants. For a pointed interpretation $(\Imc,d)$, we
write $(\Imc,d) \models \exists^\sim_{\overline \Sigma}. \Tmc$ when
$(\Imc,d)$ is $\Sigma$-bisimilar to some pointed interpretation
$(\Jmc,d')$ with $\Jmc$ a model of $\Tmc$. The notation reflects that
what we express here can be understood as a form of bisimulation
quantifier, see~\cite{French}.

 
\begin{theorem}\label{bisimuniform}
  Let $\Tmc$ be a TBox and $\Sigma\subseteq {\sf sig}(\Tmc)$. A
  $\Sigma$-TBox $\Tmc_\Sigma$ is a uniform $\Sigma$-interpolant of
  $\Tmc$ iff for all interpretations~$\Imc$,
  \begin{equation*}
\Imc \models \Tmc_\Sigma \quad \Leftrightarrow \quad \text{for all }d \in \Delta^{\Imc},\
(\Imc,d) \models \exists^\sim_{\overline{\Sigma}} . \Tmc.
\tag{$*$}
  \end{equation*}
\end{theorem}
%
For $\Imc$ of finite outdegree, one can prove this result by employing
compactness arguments and Theorem~\ref{theorem:bisimchar}. To prove it
in its full generality, we need the automata-theoretic machinery
introduced in Section~\ref{sect:automatastuff}. We illustrate
Theorem~\ref{bisimuniform} by sketching a proof of
Example~\ref{ex1}(i). Correctness of \ref{ex1}(ii) is proved in the
appendix, while \ref{ex1}(iii) and \ref{ex1}(iv) are addressed in
Section~\ref{sect:charact}. An interpretation $\Imc$ is called a
\emph{tree interpretation} if the undirected graph
$(\Delta^{\Imc},\bigcup_{r\in \NR}r^{\Imc})$ is a (possibly infinite)
tree and $r^\Imc \cap s^\Imc = \emptyset$ for all distinct $r,s \in \NR$.
%
\begin{example}
\label{ex:charact}
  {\em Let $\Tmc_1=\{A \sqsubseteq \exists r.B\sqcap \exists r.\neg B\}$, $\Sigma_1=\{A,r\}$,
    and $\Tmc_1'= \{A \sqsubseteq \exists r.\top\}$ as in Example~\ref{ex1}(i). We have
       $\Imc \models \Tmc'_1$
    \\[1mm]
   iff $\forall d \in \Delta^\Imc$: $(\Imc,d) \sim_{\Sigma_1} (\Jmc,d)$ for a tree model \Jmc of
   $\Tmc'_1$\\[1mm]
   iff $\forall d \in \Delta^\Imc$: $(\Imc,d) \sim_{\Sigma_1} (\Jmc,d)$ for a tree interpretation~\Jmc\\
   \hspace*{12mm}such that $e \in A^\Jmc$ implies $|\{d \mid (d,e) \in r^\Jmc \}| \geq 2$ \\[1mm]
   iff $\forall d \in \Delta^\Imc$: $(\Imc,d) \sim_{\Sigma_1} (\Jmc,d)$ for a tree interpretation \Jmc\\
   \hspace*{12mm}such that $e \in A^\Jmc$ implies $e \in (\exists r . B \sqcap \exists r . \neg B)^\Jmc$ \\[1mm]
   iff $\forall d \in \Delta^\Imc$: $(\Imc,d) \models \exists^\sim_{\overline{\Sigma}_{1}} .\Tmc_{1}$.

    \smallskip
    \noindent
    The first `iff' relies on the fact that unraveling
    an interpretation into a tree interpretation preserves
    bisimularity, the second one on the fact that bisimulations are
    oblivious to the duplication of successors, and the third one on a
    reinterpretation of $B \notin \Sigma_1$ in~\Jmc.
}
\end{example}
Theorem~\ref{bisimuniform} also yields a characterization of
conservative extensions in terms of bisimulations, which is as follows.
\begin{theorem}\label{bisimCE}
  Let $\Tmc,\Tmc'$ be TBoxes. Then $\Tmc \cup \Tmc'$
  is a conservative extension of \Tmc iff for all interpretations
  $\Imc$,
$
\Imc \models \Tmc \Rightarrow \text{for all }d \in \Delta^{\Imc},\
(\Imc,d) \models \exists^\sim_{\overline{\Sigma}} . \Tmc'$ where $\Sigma=\mn{sig}(\Tmc)$.
\end{theorem}

%


\section{Characterizing Existence of Interpolants}
\label{sect:charact}

If we admit TBoxes that are infinite, then uniform
$\Sigma$-interpolants always exist: for any TBox $\Tmc$ and
signature~$\Sigma$, the infinite TBox $\Tmc_\Sigma^\infty$ that
consists of all $\Sigma$-inclusions $C \sqsubseteq D$ with $\Tmc
\models C \sqsubseteq D$ is a uniform $\Sigma$-interpolant of \Tmc.
To refine this simple observation, we define the \emph{role-depth}
$\mn{rd}(C)$ of a concept $C$ to be the nesting depth of existential
restrictions in $C$.  For every finite signature $\Sigma$ and $m\geq
0$, one can fix a \emph{finite} set $\mathcal{C}_{f}^{m}(\Sigma)$ of
$\Sigma$-concepts $D$ with $\mn{rd}(D)\leq m$ such that every
$\Sigma$-concept $C$ with $\mn{rd}(C)\leq m$ is equivalent to some
$D\in \mathcal{C}_{f}^{m}(\Sigma)$. Let
\vspace*{-0.1cm}
$$ 
\Tmc_{\Sigma,m} = \{ C \sqsubseteq D \mid \mathcal{T} \models C \sqsubseteq D
\mbox{ and } C,D\in \mathcal{C}_{f}^{m}(\Sigma)\}.
\vspace*{-0.1cm}
$$
Clearly, $\Tmc^\infty_\Sigma$ is equivalent to $\bigcup_{m \geq 0}
\Tmc_{\Sigma,m}$ suggesting that if a uniform interpolant exists,
it is one of the TBoxes $\Tmc_{\Sigma,m}$.
In fact, 
it is easy to see that the following
are equivalent (this is similar to the approximation of uniform
interpolants in \cite{DBLP:conf/semweb/WangWTPA09}):
\begin{itemize}

\item[(a)] there does not exist a uniform $\Sigma$-interpolant of $\Tmc$;

\item[(b)] no $\Tmc_{\Sigma,m}$ is a uniform $\Sigma$-interpolant of $\Tmc$;

\item[(c)] for all $m\geq 0$ there is a $k>m$ such that 
$\Tmc_{\Sigma,m}\not\models \Tmc_{\Sigma,k}$.

\end{itemize}
Our characterization of the (non)-existence of uniform interpolants is
based on an analysis of the TBoxes $\Tmc_{\Sigma,m}$.  For an
interpretation \Imc, $d \in \Delta^\Imc$, and $m \geq 0$, we use
$\Imc^{\leq m}(d)$ to denote the \emph{m-segment generated by $d$ in
  $\Imc$}, i.e., the restriction of \Imc to those elements of
$\Delta^\Imc$ that can be reached from $d$ in at most $m$ steps in the
graph $(\Delta^\Imc, \bigcup_{r \in \NR} r^\Imc)$. Using the
definition of~$\Tmc_{\Sigma,m}$, Theorem~\ref{theorem:bisimchar}, and
the fact that every $m$-segment can be described up to bisimulation
using a concept of role-depth $m$, it can be shown that an
interpretation \Imc is a model of $\Tmc_{\Sigma,m}$ iff each of \Imc's
$m$-segments is $\Sigma$-bisimilar to an $m$-segment of a model of
\Tmc. Thus, if $\Tmc_{\Sigma,m}$ is \emph{not} a uniform interpolant,
then this is due to a problem that cannot be `detected' by
$m$-segments, i.e., some $\Sigma$-part of a model of \Tmc that is
located before an $m$-segment can pose constraints on $\Sigma$-parts
of the model after that segment, where `before' and `after' refer 
to reachability in $(\Delta^\Imc, \bigcup_{r \in \NR} r^\Imc)$.

The following result describes this in an exact way. Together with the
equivalence of (a) and (b) above, it yields a first characterization
of the existence of uniform interpolants.  
$\rho^{\Imc}$ denotes the root of a tree interpretation $\Imc$,
$\Imc^{\leq m}$ abbreviates $\Imc^{\leq m}(\rho^\Imc)$,
and a $\Sigma$-tree interpretation is a tree interpretations that only interprets
predicates from $\Sigma$. 
\vspace*{-0.1cm}
%
%
%
%
%
%
\begin{theorem}\label{thm1}
Let $\Tmc$ be a TBox, $\Sigma\subseteq {\sf sig}(\Tmc)$, and
$m\geq 0$. Then $\Tmc_{\Sigma,m}$ is not a uniform $\Sigma$-interpolant of 
$\Tmc$ iff 

\smallskip

\noindent $(\ast_{m})$ there exist two $\Sigma$-tree interpretations, 
$\mathcal{I}_{1}$ and $\mathcal{I}_{2}$, of finite outdegree such that
\begin{enumerate}
\item $\mathcal{I}_{1}^{\leq m} \;=\;  \mathcal{I}_{2}^{\leq m}$;
\item $(\mathcal{I}_{1},\rho^{\mathcal{I}_{1}})\models  
\exists^\sim_{\overline{\Sigma}}. \Tmc$;
\item $(\mathcal{I}_{2},\rho^{\mathcal{I}_{2}}) \not\models 
\exists^\sim_{\overline{\Sigma}}. \Tmc$;
\item For all successors $d$ of $\rho^{\mathcal{I}_{2}}$:
$(\mathcal{I}_{2},d) \models  \exists^\sim_{\overline{\Sigma}}. \Tmc$.
\end{enumerate}
\end{theorem}
%
Intuitively, Points~1 and~2 ensure that $\mathcal{I}_{1}^{\leq m}
\;=\; \mathcal{I}_{2}^{\leq m}$ is an $m$-segment of a model of
$\Tmc_{\Sigma,m}$, Points~2 and~3 express that in models of \Tmc, the
$\Sigma$-part after the $m$-segment is constrained in some way, and Point~4
says that this is due to $\rho^{\Imc_1}$ and $\rho^{\Imc_2}$, i.e.,
the constraint is imposed `before' the $m$-segment.  The following example
demonstrates how Theorem~\ref{thm1} can be used to prove non-existence
of uniform interpolants.
\vspace*{-0.1cm}
\begin{example}
{\em Let $\Tmc_3=\{A \sqsubseteq B, B \sqsubseteq \exists r.B\}$ and
$\Sigma_{3}=\{A,r\}$ as in Example~\ref{ex1}(iii). We show that $(\ast_{m})$ holds for all $m$
and thus, there is no uniform $\Sigma_3$-interpolant of $\Tmc_3$.
Example~\ref{ex1}(iv) is treated in the long version.

Let $m\geq 0$.  Set
$\mathcal{I}_{1}=(\{0,1,\ldots\},A^{\mathcal{I}_{1}},r^{\mathcal{I}_{1}})$,
where $A^{\mathcal{I}_{1}}=\{0\}$ and $r^{\mathcal{I}_{1}} = \{(n,n+1)
\mid n\geq 0\}$, and let $\mathcal{I}_{2}$ be the restriction of
$\mathcal{I}_{1}$ to $\{0,\ldots,m\}$. Then (1)~$\Imc_{1}^{\leq
  m}=\Imc_{2}^{\leq m}$;
(2)~$(\mathcal{I}_{1},0)\models\exists^{\sim}_{\overline{\Sigma}_{3}}
\Tmc_{3}$ as the expansion of $\mathcal{I}_{1}$ by
$B^{\mathcal{I}_{1}} = \{0,1,\ldots\}$ is a model of
$\mathcal{T}_{3}$; (3)~$(\mathcal{I}_{2},0)\not\models
\exists^{\sim}_{\overline{\Sigma}_{3}}\Tmc_3$ as there is no infinite
$r$-sequence in $\Imc_2$ starting at $0$; and
(4)~$(\mathcal{I}_{2},1)\models
\exists^{\sim}_{\overline{\Sigma}_{3}}\Tmc_3$ as the restriction of
$\Imc_2$ to $\{1,\dots,m\}$ is a model of $\Tmc_3$.  }
\end{example}
The next example illustrates another use of Theorem~\ref{thm1} by
identifying a class of signatures for which uniform interpolants
always exist. Details are given in the long version.

\begin{example}[Forgetting stratified concept names] 
\label{ex:forget}{\em 
  A concept name $A$ is \emph{stratified in $\Tmc$} if all occurrences of $A$ 
  in concepts from $\mn{conc}(\Tmc)=\{C,D
  \mid C \sqsubseteq D \in \Tmc \}$ are exactly in nesting depth $n$
  of existential restrictions, for some $n\geq 0$. Let $\Tmc$ be a TBox and $\Sigma$ a
  signature such that $\mn{sig}(\Tmc) \setminus \Sigma$ consists of
  stratified concept names only, i.e., we want to \emph{forget} a set of
  stratified concept names. Then the existence of a uniform
  $\Sigma$-interpolant of \Tmc is guaranteed; moreover,
  $\Tmc_{\Sigma,m}$ is such an interpolant, where $m=\max\{ {\sf
    rd}(C) \mid C \in \mn{conc}(\Tmc) \}$.  }
\end{example} 
To turn Theorem~\ref{thm1} into a decision procedure for the existence
of uniform interpolants, we prove that rather than testing $(\ast_m)$
for all $m$, it suffices to consider a single number $m$. This yields
the final characterization of the existence of uniform interpolants.
We use $|\Tmc|$ to denote the \emph{length} of a TBox \Tmc, i.e., the
number of symbols needed to write it.
\begin{theorem}\label{fixm}
  Let $\Tmc$ be a TBox and $\Sigma\subseteq {\sf sig}(\Tmc)$. Then
  there does not exist a uniform $\Sigma$-interpolant of $\Tmc$
\vspace*{-1mm}
  iff $(\ast_{M_{\Tmc}^{2}+1})$ from Theorem~\ref{thm1} holds, where
  $M_{\Tmc}:=2^{2^{|\Tmc|}}$.
\end{theorem}
It suffices to show that $(\ast_{M_{\Tmc}^{2}+1})$ implies
$(\ast_{m})$ for all $m\geq M_{\Tmc}^{2}+1$.  The proof idea
is as follows.  Denote by $\mn{cl}(\Tmc)$ the closure under single
negation and subconcepts of $\mn{conc}(\Tmc)$.
%
%
%
%
%
%
%
%
%
%
The \emph{type} of some $d\in \Delta^{\Imc}$ in an interpretation
$\Imc$ is 
$$ \mn{tp}^\Imc(d) := \{ C\in {\mn{cl}(\Tmc)} \mid d \in C^\Imc\}.  $$
Many constructions for $\ALC$ (such as blocking in tableaux,
filtrations of interpretations, etc.)  exploit the fact that the
relevant information about any element $d$ in an interpretation is
given by its type. This can be exploited e.g.\ to prove \ExpTime upper
bounds as there are `only' exponentially many distinct types.  In the
proof of Theorem~\ref{fixm}, we make use of a `pumping lemma' that
enables us to transform any pair $\Imc_{1},\Imc_{2}$ witnessing
$(\ast_{M_{\Tmc}^{2}+1})$ into a witness $\Imc_{1}',\Imc_{2}'$ for
$(\ast_{m})$ when $m\geq M_{\Tmc}^{2}+1$. The construction depends on
the relevant information about elements of $\Delta^{\Imc_1}$ and
$\Delta^{\Imc_2}$; in contrast to standard constructions, however,
types are not sufficient and must be replaced by \emph{extension sets}
${\sf Ext}^{\Imc}(d)$, defined as
$$
{\sf Ext}^{\Imc}(d) = \{ \mn{tp}^{\mathcal{I}}(d') \mid \exists \Jmc:
\mathcal{J}\models \mathcal{T} \text{ and } 
(\mathcal{I},d) \sim_{\Sigma} (\mathcal{J},d')\}
$$
and capturing all ways in which the restiction of $\mn{tp}^\Imc(d)$ to
$\Sigma$-concepts can be extended to a full type in models of
$\Tmc$. As the number of such extension sets is double exponential in
$|\Tmc|$ and we have to consider pairs $(d_1,d_2) \in
\Delta^{\Imc_1} \times \Delta^{\Imc_2}$, we (roughly) obtain a
$M_{\Tmc}^{2}$ bound. Details are in the long version.

%
%
%
%

We note that, by Theorem~\ref{fixm}, the uniform
$\Sigma$-interpolant of a TBox \Tmc exists iff $\Tmc\cup \Tmc_{\Sigma,M_\Tmc^{2}+1}$ is a
conservative extension of $\Tmc_{\Sigma,M_{\Tmc}^{2}+1}$. With the
decidability of conservative extensions proved in
\cite{KR06}, this yields decidablity of the existence of uniform
interpolants. However, the size of $\Tmc_{\Sigma,M_{\Tmc}^{2}+1}$ is
non-elementary, and so is the running time of the resulting algorithm.
We next show how to improve this.
\vspace*{-0.2cm}

\section{Automata Constructions / Complexity}
\label{sect:automatastuff}

We develop a worst-case optimal algorithm for deciding the existence
of uniform interpolants in \ALC, exploiting Theorem~\ref{fixm} and
making use of alternating automata. As a by-product, we prove the
fundamental characterization of uniform interpolants in terms of
bisimulation stated as Theorem~\ref{bisimuniform} without the initial
restriction to interpretations of finite outdegree. We also obtain a
representation of uniform interpolants as automata and a novel, more
transparent proof of the 2-\ExpTime upper bound for deciding
conservative extensions originally established in~\cite{KR06}.

We use amorphous alternating parity tree automata in the style of
Wilke \cite{Wilke-Automata}, which
run
on unrestricted interpretations rather than on trees, only. We 
call them \emph{tree} automata as they are in the tradition of more
classical forms of such automata. In particular, a run of an automaton
is tree-shaped, even if the input interpretation is not.
\begin{definition}[APTA]
\label{def:apta}
  An \emph{alternating parity tree automaton (APTA)} is a tuple $\Amc
  = (Q,\Sigma_N,\Sigma_E,q_0, \delta, \Omega)$, where $Q$ is a finite
  set of \emph{states}, $\Sigma_N \subseteq \NC$ is the finite
  \emph{node alphabet}, $\Sigma_E \subseteq \NR$ is the finite
  \emph{edge alphabet}, $q_0 \in Q$ is the \emph{initial state},
  $
    \delta: Q  \rightarrow \mn{mov}(\Amc),
  $
  is the transition function with
  $
    \mn{mov}(\Amc)=\{ \mn{true}, \mn{false}, A, \neg A, q, q \wedge q',
    q \vee q', \langle r \rangle q, [r] q \mid 
     A \in \Sigma_N, q,q' \in Q, r \in \Sigma_E \}
  $
  the set of \emph{moves} of the automaton, and $\Omega:Q
  \rightarrow \Nbbm$ is the \emph{priority function}.
\end{definition}
Intuitively, the move $q$ means that the automaton sends a copy of
itself in state $q$ to the element of the interpretation that it is
currently processing, $\auf r \zu q$ means that a copy in state $q$ is
sent to an $r$-successor of the current element, and $[r] q$ means
that a copy in state $q$ is sent to every $r$-successor.

It will be convenient to use arbitrary modal formulas in
negation normal form when specifying the transition function of
APTAs. The more restricted form required by Definition~\ref{def:apta}
can then be attained by introducing intermediate
states. In subsequent constructions that involve APTAs, we will not
describe those additional states explicitly.  However, we will
(silently) take them into account when stating size bounds for
automata.


In what follows, a \emph{$\Sigma$-labelled tree} is a pair $(T,\ell)$
with $T$ a tree and $\ell:T \rightarrow \Sigma$ a node
labelling function. A \emph{path} $\pi$ in a tree $T$ is a subset of
$T$ such that $\varepsilon \in \pi$ and for each $x \in \pi$ that is
not a leaf in $T$, $\pi$ contains one son of $x$.
\begin{definition}[Run]
\label{def:altrun}
Let $(\Imc,d_0)$ be a pointed $\Sigma_N \cup \Sigma_E$-interpretation
and $\Amc= (Q,\Sigma_N,\Sigma_E,q_0, \delta, \Omega)$ an APTA. A
\emph{run} of \Amc on $(\Imc ,d_0)$ is a $Q \times \Delta^\Imc$-labelled tree
$(T,\ell)$ such that $\ell(\varepsilon)=(q_0,d_0)$ and for every
$x \in T$ with $\ell(x)=(q,d)$:
  \begin{itemize}

  \item $\delta(q) \neq \mn{false}$;
    
  \item if $\delta(q) = A$ ($\delta(q) = \neg A$), then $d \in A^\Imc$ ($d \notin A^\Imc$);

  \item if $\delta(q) = q' \wedge q''$, then there are
    sons $y,y'$ of $x$ with $\ell(y)=(q',d)$ and $\ell(y')=(q'',d)$;

  \item if $\delta(q) = q' \vee q''$, then there is
    a son $y$ of $x$ with $\ell(y)=(q',d)$ or $\ell(y')=(q'',d)$;

  \item if $\delta(q) =\langle r \rangle q'$, then
    there is a $(d,d') \in r^\Imc$ and a son $y$ of $x$ with
    $\ell(y)=(q',d')$;

  \item if $\delta(q) =[ r ] q'$ and $(d,d') \in
    r^\Imc$, then there is a son $y$ of $x$ with $\ell(y)=(q',d')$.

  \end{itemize}
  A run $(T,\ell)$ is \emph{accepting} if for every path $\pi$ of $T$,
  the maximal $i \in \Nbbm$ with $\{ x \in \pi \mid \ell(x)=(q,d)
  \text{ with } \Omega(q)=i \}$ infinite is even.  We use $L(\Amc)$ to
  denote the language accepted by \Amc, i.e., the set of pointed $\Sigma_N
  \cup \Sigma_E$-interpretations $(\Imc,d)$ such that there is an
  accepting run of \Amc on $(\Imc,d)$.
\end{definition}
%
Using the fact that runs are always tree-shaped, it is easy to prove
that the languages accepted by APTAs are closed under ${\Sigma_N \cup
  \Sigma_E}$-bisimulations. It is this property that makes this
automaton model particularly useful for our purposes. APTAs can be
complemented in polytime in the same way as other alternating tree
automata, and for all APTAs $\Amc_1$ and $\Amc_2$, one can construct
in polytime an APTA that accepts $L(\Amc_1) \cap L(\Amc_2)$. Wilke
shows that the emptiness problem for APTAs is \ExpTime-complete
\cite{Wilke-Automata}.

We now show that uniform $\Sigma$-interpolants of \ALC-TBoxes can be \emph{represented as
APTAs}, in the sense of the following theorem and
of Theorem~\ref{bisimuniform}.
\begin{theorem}
\label{thm:bisimauto}
Let \Tmc be a TBox and $\Sigma \subseteq \mn{sig}(\Tmc)$ a
signature. Then there exists an APTA $\Amc_{\Tmc,\Sigma}= (Q,\Sigma
\cap \NC,\Sigma \cap \NR,q_0, \delta, \Omega)$ with $|Q| \in
2^{\Omc(|\Tmc|)}$ such that $L(\Amc_{\Tmc,\Sigma})$ consists of all
\vspace*{-0.5mm}
pointed $\Sigma$-interpretations $(\Imc,d)$ with $(\Imc,d)
\models \exists^\sim_{\overline\Sigma}. \Tmc$. 
$\Amc_{\Tmc,\Sigma}$ can be constructed in time $2^{p(|\Tmc|)}$, $p$ a
polynomial.
\end{theorem}
The construction of the automaton $\Amc_{\Tmc,\Sigma}$ from
Theorem~\ref{thm:bisimauto} resembles the construction of uniform
interpolants in the  $\mu$-calculus using non-deterministic
automata described in \cite{Hollenberg-Agostino}, but is transferred to
TBoxes and alternating automata.  

Fix a TBox \Tmc and a signature $\Sigma$ and assume w.l.o.g.\ that
\Tmc has the form $\{ \top \sqsubseteq C_\Tmc \}$, with $C_\Tmc$ in
negation normal form \cite{Baader-et-al-03b}.
%
%
%
%
%
%
%
Recall the notion of a type introduced in Section~\ref{sect:charact}.
We use $\mn{TP}(\Tmc)$ to denote the set of all types realized in
some model of \Tmc, i.e.,
$
  \mn{TP}(\Tmc)= \{ \mn{tp}^\Imc(d) \mid \Imc \text{ model of } \Tmc,
  d \in \Delta^\Imc \}.
$
Note that $\mn{TP}(\Tmc)$ can be computed in time exponential in the
size of \Tmc since concept satisfiability w.r.t.\ TBoxes is
\ExpTime-complete in \ALC \cite{Baader-et-al-03b}.  Given $t,t' \in
\mn{TP}(\Tmc)$ and $r\in \Sigma$, we write $t \leadsto_{r} t'$ if
$C\in t'$ implies $\exists r . C\in t$ for all $\exists r.C\in
{\mn{cl}(\Tmc)}$.  Now define the automaton $\Amc_{\Tmc,\Sigma} :=
(Q,\Sigma_N,\Sigma_E,q_0, \delta, \Omega)$, where
$$
\begin{array}{r@{\;}c@{\;}l}
Q &=& \mn{TP}(\Tmc) \uplus \{ q_0 \} \quad
\Sigma_N = \Sigma \cap \NC \quad
\Sigma_E = \Sigma \cap \NR \\[1mm]
\delta(q_0)&=&\displaystyle \bigvee \mn{TP}(\Tmc) 
\\[1mm]
\delta(t)&=&\displaystyle \bigwedge_{A \in t \cap \NC \cap \Sigma} A
\wedge \bigwedge_{A \in (\NC \cap \Sigma)\setminus t} \neg A \\[5mm]
&&\wedge \; \displaystyle  \bigwedge_{r \in \Sigma \cap \NR} \!
  [r] \bigvee \{ t' \in \mn{TP}(\Tmc) \mid t \leadsto_{r} t'\} 
\\[5mm]
&& \wedge \; \displaystyle
\bigwedge_{\exists r . C \in t, r \in \Sigma} \!\!\!\!\!\! \langle r \rangle \bigvee \{ t' \in \mn{TP}(\Tmc) \mid C \in t \wedge t \leadsto_{r} t' \} 
\\[1mm]
\Omega(q)&=&0 \text{ for all } q  \in Q
\end{array}
$$
Here, the empty conjunction represents \mn{true} and the empty disjunction
represents \mn{false}.  The acceptance condition of the automaton is
trivial, which (potentially) changes when we complement it
subsequently.
We prove in the appendix that this automaton satisfies the
conditions in Theorem~\ref{thm:bisimauto}. 

\medskip We now develop a decision procedure for the existence of
uniform interpolants by showing that 
%
%
the characterization of the existence of uniform interpolants
provided by Theorem~\ref{fixm} can be captured by APTAs, in the following
sense.
\begin{theorem}
\label{prop:cetocons}
Let \Tmc be a TBox, $\Sigma \subseteq \mn{sig}(\Tmc)$ a signature, and
$m \geq 0$. Then there is an APTA
$\Amc_{\Tmc,\Sigma,m}=(Q,\Sigma_N,\Sigma_E,q_0, \delta, \Omega)$ such
that $L(\Amc) \neq \emptyset$ iff Condition~($*_m$) from
Theorem~\ref{thm1} is satisfied.  
%
%
%
%
%
Moreover, $|Q| \in \Omc(2^{\Omc(n)}+\log^2 m)$ and $|\Sigma_N|,|\Sigma_E| \in \Omc(n+\log m)$,
where $n=|\Tmc|$.
\end{theorem}
The size of $\Amc_{\Tmc,\Sigma,m}$ is exponential in $|\Tmc|$ and
logarithmic in~$m$.
By Theorem~\ref{fixm}, we can set
$m=2^{2^{|\Tmc|}}$, and thus the size of $\Amc_{\Tmc,\Sigma,m}$ is
exponential in~$|\Tmc|$.  Together with the \ExpTime emptiness test for
APTAs, we obtain a 2-\ExpTime decision procedure for the existence of
uniform interpolants.  We construct $\Amc_{\Tmc,\Sigma,m}$ as an
intersection of four APTAs, each ensuring one of the conditions
of ($*_m$); building the automaton for Condition~2 involves
complementation. The automaton
$\Amc_{\Tmc,\Sigma,m}$ runs over an extended alphabet that allows to
encode both of the interpretations $\Imc_1$ and $\Imc_2$ mentioned in
($*_m$), plus a `depth counter' for enforcing Condition~1 of
($*_m$).

\medskip

A similar, but simpler construction can be used to reprove the
2-\ExpTime upper bound for deciding conservative extensions
established in \cite{KR06}. The construction only depends on
Theorem~\ref{bisimCE}, but not on the material in
Section~\ref{sect:charact} and is arguably more transparent than
the original one.
\begin{theorem}
  \label{thm:ceupper}
  Given TBoxes \Tmc and $\Tmc'$, it can be decided in time
  $2^{p(|\Tmc| \cdot 2^{|\Tmc'|})}$ whether $\Tmc \cup \Tmc'$ is a
  conservative extension of $\Tmc$, for some polynomial $p()$.
\end{theorem}
A 2-\ExpTime lower bound was also established in \cite{KR06}, thus the
upper bound stated in Theorem~\ref{thm:ceupper} is tight. This lower
bound transfers to the existence of uniform interpolants: one can
show that $\Tmc' = \Tmc \cup \{\top \sqsubseteq C\}$ is a
conservative extension of \Tmc iff there is a uniform ${\sf sig}(\Tmc)\cup \{r\}$-interpolant 
of 
$$
\Tmc_{0} = \Tmc \cup \{ \neg C \sqsubseteq A, A\sqsubseteq \exists r.A\} \cup 
\{\exists s.A \sqsubseteq A\mid s\in {\sf sig}(\Tmc')\},
$$
with $r,A$ are fresh. This yields the main result of this
section.
\begin{theorem}
  \label{thm:unifint2expupper}
  It is 2-\ExpTime-complete to decide, given a TBox \Tmc and
  a signature $\Sigma \subseteq \mn{sig}(\Tmc)$, whether there exists a
  uniform $\Sigma$-interpolant of \Tmc.
\end{theorem}


\section{Computing Interpolants / Interpolant Size}
\label{sect:lower}

We show how to compute smaller uniform interpolants than 
the non-elementary $\Tmc_{\Sigma,M_{\Tmc}^{2}+1}$ and establish a matching upper bound
on their size. Let $C$ be a concept and
$\Sigma\subseteq {\sf sig}(C)$ a signature. A concept $C'$ is called a
\emph{concept uniform $\Sigma$-interpolant of $C$} if ${\sf sig}(C)
\subseteq \Sigma$, $\emptyset \models C \sqsubseteq C'$, and
$\emptyset \models C' \sqsubseteq D$ for every concept $D$ such that
${\sf sig}(D) \subseteq \Sigma$ and $\emptyset
\models C \sqsubseteq D$.  The following result is proved in
\cite{CCMV-KR06}.
\begin{theorem}\label{conceptuniform}
For every concept $C$ and signature $\Sigma\subseteq {\sf sig}(C)$
one can effectively compute a concept uniform $\Sigma$-interpolant $C'$ of 
$C$ of at most exponential size in $C$.
\end{theorem}
This result can be lifted to (TBox)
uniform interpolants by `internalization' of the TBox. This is very
similar to what is attempted in \cite{DBLP:conf/ecai/WangWTZ10}, but
we use different bounds on the role depth of the internalization
concepts. More specifically, let $\Tmc= \{\top \sqsubseteq C_{\Tmc}\}$
have a uniform $\Sigma$-interpolant and $R$ denote the set of role
names in $\Tmc$. For a concept $C$, define inductively
$$
\forall R^{\leq 0}.C = C,\quad \forall R^{\leq n+1}.C = C \sqcap \bigsqcap_{r\in R}\forall r.
\forall R^{\leq n}.C
$$ 
It can be shown using Theorem~\ref{fixm} that for $m = 2^{2^{|C_\Tmc|+1}} + 2^{|C_\Tmc|} +2$
and $C$ a concept uniform $\Sigma$-interpolant of $\forall R^{\leq
  m}.C_\Tmc$, the TBox $\Tmc'= \{ \top \sqsubseteq
C\}$ is a uniform $\Sigma$-interpolant of $\Tmc$. A close
inspection of the construction underlying the proof of
Theorem~\ref{conceptuniform} applied to $\forall R^{\leq m}.C_\Tmc$
reveals that ${\sf rd}(C)\leq {\sf rd}(\forall R^{\leq m}.C_{\Tmc})$ and that the size of $C$ is
at most triple exponential in $|\Tmc|$. This yields the following upper bound.
\begin{theorem}
\label{thm:internalization}
Let \Tmc be an \ALC-TBox and $\Sigma\subseteq {\sf sig}(\Tmc)$. If there is a uniform
$\Sigma$-interpolant of \Tmc, then there is one of size at most $2^{2^{2^{p(|\Tmc|)}}}$,
$p$ a polynomial.
\end{theorem}
A matching lower bound on the size of uniform interpolants can be obtained
by transferring a lower bound on the size of so-called witness concepts for
(non-)conservative extensions established in \cite{KR06}:
\begin{theorem}
\label{thm:trplexpsize}
There exists a signature $\Sigma$ of cardinality $4$
and a family of TBoxes $(\Tmc_n)_{n  > 0 }$ such that, for all $n > 0$,
\begin{itemize}


\item[(i)] $|\Tmc_n| \in \Omc(n^2)$ and

\item[(ii)] every uniform $\Sigma$-interpolant $\{ \top \sqsubseteq
  C_\Tmc \}$ for $\Tmc_n$ is of size at least $2^{(2^n \cdot
    2^{2^n})-2}$.

\end{itemize}
\end{theorem}
%


\section{Conclusions}

We view the characterizations, tools, and results obtained in this
paper as a general foundation for working with uniform interpolants in
expressive DLs. In fact, we believe that the established framework can
be extended to other expressive DLs such as \ALC extended with number
restrictions and/or inverse roles without too many hassles: the main
modifications required should be a suitable modification of the notion
of bisimulation and (at least in the case of number restrictions) a
corresponding extension of the automata model. Other extensions,
such as with nominals, require more efforts.

\enlargethispage*{8mm}

In concrete applications, what to do when the desired uniform
$\Sigma$-interpolant does not exist? In applications such as
\emph{ontology re-use} and \emph{ontology summary}, one option is to
extend the signature $\Sigma$, preferably in a minimal way, and then
to use the interpolant for the extended signature. We believe that
Theorem~\ref{thm1} can be helpful to investigate this further, loosely
in the spirit of Example~\ref{ex:forget}. In applications such as
\emph{predicate hiding}, an extension of $\Sigma$ might not be
acceptable. It is then possible to resort to a more expressive DL in
which uniform interpolants always exist. In fact,
Theorem~\ref{thm:bisimauto} and the fact that APTAs have the same
expressive power as the $\mu$-calculus \cite{Wilke-Automata} point the
way towards the extension of \ALC with fixpoint operators.

\medskip
\noindent
{\bf Acknowledgments.}\ \ This work was supported by the DFG SFB/TR 8 ``Spatial Cognition''.


\cleardoublepage

\appendix

\section{Proofs for Section~\ref{sect:prelim}}

{\bf Proof sketch for Example~\ref{ex1} (ii)}
Recall that $\Tmc_{2}=\{A \equiv B \sqcap \exists r.B\}$ and $\Sigma_{2}=\{A,r\}$.
We show that $\Tmc'_{2}=\{A \sqsubseteq \exists r.(A \sqcup \neg \exists r.A)\}$
is a uniform $\Sigma_{2}$-interpolant of~$\Tmc_2$.
To this end, we prove the criterion of Theorem~\ref{bisimuniform}.

Let $\Imc$ be a model of $A \sqsubseteq \exists r.(A \sqcup \neg \exists r.A)$.
Then every $(\Imc,d)$, $d\in \Delta^{\Imc}$, is bisimilar to a tree-interpretation
$(\Jmc,\rho^{\Jmc})$ that is a model of $A \sqsubseteq \exists r.(A \sqcup \neg \exists r.A)$.
We define a new interpretation $\Jmc'$ that coincides with $\Jmc$ except that $B$ is interpreted as follows:
for every $e\in \Delta^{\Jmc}$ with $e\in A^{\Jmc}$ let 
$e\in B^{\Jmc'}$ and, if there does not exist an $r$-successor of $e$ in $A^{\Jmc}$, then take an $r$-successor
$e'$ of $e$ with $e'\not\in (\exists r.A)^{\Jmc}$ and let $e'\in B^{\Jmc'}$ as well. Such an $e'$ exists since
$\Jmc$ is a model of $A \sqsubseteq \exists r.(A \sqcup \neg \exists r.A)$. It is readily checked that
$\Jmc'$ is a model of $\Tmc_{2}$. Thus, for all $d\in \Delta^{\Imc}$, 
$(\Imc,d)\models \exists^{\sim}_{\overline{\Sigma}_{2}}.\Tmc_{2}$.

Conversely, let $(\Imc,d)\models \exists^{\sim}_{\overline{\Sigma}_{2}}.\Tmc_{2}$, for all $d\in \Delta^{\Imc}$.
Let $d\in A^{\Imc}$ and assume $d\not\in \exists r.(A \sqcup \neg \exists r.A)$.
Then, no $r$-successor of $d$ is in $A^{\Imc}$ and all $r$-successors of $d$ have an $r$-successor in $A^{\Imc}$.
Let $(\Imc,d)\sim_{\Sigma_{2}} (\Jmc,d')$ with $\Jmc$ a model of $\Tmc_{2}$.
Then $d'\in A^{\Jmc}$, no $r$-successor of $d'$ is in $A^{\Jmc}$, and all $r$-successors of $d'$
have an $r$-successor in $A^{\Jmc}$.
We have $d'\in (B\sqcap \exists r.B)^{\Jmc}$ and so $d'\in B^{\Jmc}$ and there exists an $r$-successor 
$d''$ of $d'$ such that $d''\in B^{\Jmc}$.
Since $d''\not\in A^{\Jmc}$, there does not exist an $r$-successor of $d''$ that is in $B^{\Jmc}$.
But then no $r$-successor of $d''$ is in $A^{\Jmc}$ and we have derived a contradiction.

\medskip

To prove Theorem~\ref{bisimuniform} in its full generality, we will rely on the
automata-theoretic machinery introduced in
Section~\ref{sect:automatastuff}. For now, we only establish a
modified version where ``for all interpretation \Imc'' is replaced
with ``for all interpretations \Imc of finite outdegree''.

\noindent
{\bf Theorem~\ref{bisimuniform} (Modified Version)}
  Let $\Tmc$ be a TBox and $\Sigma\subseteq {\sf sig}(\Tmc)$. A
  $\Sigma$-TBox $\Tmc_\Sigma$ is a uniform $\Sigma$-interpolant of $\Tmc$
  iff for all interpretations~$\Imc$ of finite outdegree,
  \begin{equation*}
\Imc \models \Tmc_\Sigma \quad \Leftrightarrow \quad \text{for all }d \in \Delta^{\Imc},\
(\Imc,d) \models \exists^\sim_{\overline{\Sigma}} . \Tmc.
\tag{$*$}
  \end{equation*}
\begin{proof}
``if''. Assume that ($*$) is satisfied for all interpretations \Imc
   of finite outdegree. We have to show that for all
   $\Sigma$-inclusions $C \sqsubseteq D$,
   $$
     \Tmc_\Sigma \models C \sqsubseteq D \text{ iff } \Tmc \models C \sqsubseteq D.
   $$
   Let $\mn{sig}(C \sqsubseteq D) \subseteq \Sigma$ and assume
   $\Tmc_\Sigma \not\models C \sqsubseteq D$. Then $C \sqcap \neg D$
   is satisfiable w.r.t.\ $\Tmc_\Sigma$, i.e., there is a pointed
   model $(\Imc,d)$ of $\Tmc_\Sigma$ with finite outdegree and $d \in
   (C \sqcap \neg D)^\Imc$. By ($*$), there is thus a pointed model
   $(\Jmc,e)$ of $\Tmc$ with $(\Imc,d) \sim_\Sigma (\Jmc,e)$.
   Together with $d \in (C \sqcap \neg D)^\Imc$, the latter implies $e
   \in (C \sqcap \neg D)^\Jmc$. Thus, $C \sqcap \neg D$ is satisfiable
   w.r.t.\ $\Tmc$, implying $\Tmc \not\models C \sqsubseteq
   D$. Conversely, let $\Tmc \not\models C \sqsubseteq D$. Then there
   is a pointed model ($\Imc,d)$ of \Tmc with finite outdegree and $d
   \in (C \sqcap \neg D)^\Imc$. Trivially, \Imc satisfies the right-hand
   side of ($*$), whence $\Imc \models \Tmc_\Sigma$ and we are done.

  \smallskip
  \noindent
  
  For the ``only if '' direction, we first need a preliminary. An
  interpretation \Imc is \emph{modally saturated} iff it satisfies the
  following condition, for all $r \in \NR$: if $d \in \Delta^\Imc$ and
  $\Gamma$ is a (potentially infinite) set of concepts such that, for
  all finite $\Psi \subseteq \Gamma$, there is a $d'$ with $(d,d') \in
  r^\Imc$ and $d' \in \Psi^\Imc$, then there is an $e$ with $(d,e) \in
  r^\Imc$ and $e \in \Gamma^\Imc$. The most important facts about
  modally saturated interpretations we need here are (i)~every
  (finite) or infinite set of concepts that is satisfiable w.r.t.\ a
  TBox \Tmc is satisfiable in a modally saturated model of \Tmc; (ii)~every inte
rpretation with finite outdegree is modally saturated; and
  (iii)~Point~2 of Theorem~\ref{theorem:bisimchar} can be generalized from
  interpretations of finite outdegree to modally saturated
  interpretations \cite{GorankoOtto}.

  \smallskip

  ``only if''. Assume that $\Tmc_\Sigma$ is a uniform $\Sigma$-interpolant 
  of \Tmc. First assume that \Imc is an interpretation
  that satisfies the right-hand side of ($*$). By Point~1 of
  Theorem~\ref{theorem:bisimchar}, $\Imc \models \Tmc$ and since $\Tmc
  \models \Tmc_\Sigma$, also $\Imc \models \Tmc_\Sigma$. Now assume
  that \Imc is a model of $\Tmc_\Sigma$ of finite outdegree and let $d
  \in \Delta^\Imc$. Define $\Gamma$ to be the set of all
  $\Sigma$-concepts $C$ with $d \in C^\Imc$. Clearly, every finite
  subset $\Gamma' \subseteq \Gamma$ is satisfiable w.r.t.\
  $\Tmc_\Sigma$. Since $\Tmc_\Sigma$ is a uniform $\Sigma$-interpolant
  of $\Tmc$, every such $\Gamma'$ is also satisfiable w.r.t.\ $\Tmc$
  (since $\Tmc_\Sigma \not\models \top \sqsubseteq \neg \bigsqcap
  \Gamma'$ implies $\Tmc \not\models \top \sqsubseteq \neg \bigsqcap
  \Gamma'$). By compactness of \ALC, $\Gamma$ is satisfiable w.r.t.\
  $\Tmc$. By~(i), there thus exists a modally saturated pointed model
  $(\Jmc,e)$ of $\Tmc$ such that $e \in \Gamma^\Jmc$. Since $d \in
  \Gamma^\Imc$ and $e \in \Gamma^\Jmc$ (and since, by definition,
  $\Gamma$ contains each $\Sigma$-concept or its negation), we have $d
  \in C^\Imc$ iff $e \in C^\Jmc$ for all \ALC-concepts $C$
  over~$\Sigma$. By (ii) and (iii), this yields $\Imc \sim_\Sigma
  \Jmc$ and we are done.  \qed
\end{proof}

\section{Proofs for Section~\ref{sect:charact}}

For a tree interpretation $\Imc$ and $d\in
\Delta^{\Imc}$ we denote by $\Imc(d)$ the tree interpretation induced
by the subtree generated by $d$ in $\Imc$.

Besides of $\ALC_{\Sigma}$-equivalence, we now also require a characterization of
$\ALC_{\Sigma}$-equivalence for concepts of roles depth bounded by some $m$. 

Two pointed interpretations are
\emph{$\ALC_{\Sigma}$-m-equivalent}, in symbols 
$(\Imc_{1},d_{1})\equiv_{\Sigma}^{m}(\Imc_{2},d_{2})$,
if, and only if, for all $\Sigma$-concepts $C$ with ${\sf rd}(C)\leq m$, 
$d_{1}\in C^{\Imc_{1}}$ iff $d_{2}\in C^{\Imc_{2}}$.

The corresponding model-theoretic notion is that of $m$-bisimilarity, which is defined
inductively as follows: $(\Imc_{1},d_{1})$ and $(\Imc_{2},d_{2})$ are
\begin{itemize}
\item \emph{$(\Sigma,0)$-bisimilar},
in symbols $(\Imc_{1},d_{1}) \sim_{\Sigma}^{0}(\Imc_{2},d_{2})$, if $d_{1}\in A^{\Imc_{1}}$
iff $d_{2} \in A^{\Imc_{2}}$ for all $A \in \Sigma\cap \NC$.
\item \emph{$(\Sigma,n+1)$-bisimilar},
in symbols $(\Imc_{1},d_{1}) \sim_{\Sigma}^{n+1}(\Imc_{2},d_{2})$, if 
$(\Imc_{1},d_{1}) \sim_{\Sigma}^{0}(\Imc_{2},d_{2})$ and
\begin{itemize}
\item for all $(d_{1},e_{1}) \in r^{\Imc_{1}}$ there exists $e_{2}\in \Delta^{\mathcal{I}_{2}}$ 
such that $(d_{2},e_{2}) \in 
r^{\Imc_{2}}$ and $(\Imc_{1},e_{1}) \sim_{\Sigma}^{n} (\Imc_{2},e_{2})$, for all $r \in \Sigma$;
\item for all $(d_{2},e_{2}) \in r^{\Imc_{2}}$ there exists $e_{1} \in \Delta^{{\mathcal{I}}_{1}}$
such that $(d_{1},e_{1}) 
\in r^{\Imc_{1}}$ and $(\Imc_{1},e_{1}) \sim_{\Sigma}^{n} (\Imc_{2},e_{2})$, for all 
$r \in \Sigma$.
\end{itemize}
\end{itemize}
The following characterization is straightforward to prove and can be 
found in \cite{GorankoOtto}.
\begin{lemma}\label{lem1}
For all pointed interpretations $(\mathcal{I}_{1},d_{1})$ and $(\mathcal{I}_{2},d_{2})$,
all finite signatures $\Sigma$, and all $m\geq 0$:
$(\Imc_{1},d_{1})\equiv_{\Sigma}^{m}(\Imc_{2},d_{2})$ if, and only if,
$(\Imc_{1},d_{1})\sim_{\Sigma}^{m}(\Imc_{2},d_{2})$.
\end{lemma}

%
%
%

%


\begin{lemma}\label{lem2}
Let $(\Imc_{1},d_{1})$ and $(\Imc_{2},d_{2})$ be pointed $\Sigma$-interpretations
such that $(\mathcal{I}_{1},d_{1}) \sim_{\Sigma}^{m} (\mathcal{I}_{2},d_{2})$.
Then there exist $\Sigma$-tree interpretations $\mathcal{J}_{1}$
and $\mathcal{J}_{2}$ such that 
\begin{itemize}
\item $(\mathcal{J}_{1},\rho^{\Jmc_{1}})\sim_{\Sigma} (\mathcal{I}_{1},d_{1})$;
\item $(\mathcal{J}_{2},\rho^{\Jmc_{2}})\sim_{\Sigma} (\mathcal{I}_{2},d_{2})$; 
\item $\mathcal{J}_{1}^{\leq m} =  \mathcal{J}_{2}^{\leq m}$.
\end{itemize}
Moreover, if $\mathcal{I}_{1}$ and $\mathcal{I}_{2}$ have finite outdegree, then 
one can find such $\mathcal{J}_{1}$ and $\mathcal{J}_{2}$ that have 
finite outdegree.
\end{lemma}
\begin{proof}
Assume $(\Imc_{1},d_{1})$ and $(\Imc_{2},d_{2})$ are given.
For $j=1,2$, we define $\Jmc_{j}$ as follows.
The domain $\Delta^{\Jmc_{j}}$ consists of all words
$$
v_{0}r_{0}v_{1}r_{1}\cdots v_{n}
$$
such that
\begin{itemize}
\item $v_{0}= (d_{1},d_{2})$;
\item for all $i\leq m$ there are $e_{i},f_{i}$ such that $v_{i}= (e_{i},f_{i})$ and 
$(\Imc_{1},e_{i}) \sim_{\Sigma}^{m-i} (\Imc_{2},f_{i})$;
\item for all $i<m$: if $v_{i}= (e_{i},f_{i})$ and $v_{i+1}= (e_{i+1},f_{i+1})$,
then $(e_{i},e_{i+1})\in r_{i}^{\Imc_{1}}$ and $(f_{i},f_{i+1}) \in r_{i}^{\Imc_{2}}$;
\item for all $i > m$: $v_{i}\in \Delta^{\Jmc_{j}}$; 
\item for all $v_{m}=(e_{m},f_{m})$: if $j=1$ then $(e_{m},v_{m+1})\in r_{m}^{\Imc_{1}}$; and 
if $j=2$ then $(f_{m},v_{m+1})\in r_{m}^{\Imc_{2}}$;
\item for all $i>m$: $(v_{i},v_{i+1})\in r_{i}^{\Jmc_{j}}$.
\end{itemize}
For all concept names $A$, we set
\begin{eqnarray*}
A^{\Jmc_{j}} & = & \{ v_{0}\cdots v_{n}\in \Delta^{\Jmc_{1}}\mid n> m, v_{n}\in A^{\Imc_{j}}\}
\cup \\
& & \{(e_{0},f_{0})\cdots (e_{n},f_{n})\in \Delta^{\Jmc_{1}}\mid n\leq m,e_{n}\in A^{\Imc_{1}}\},
\end{eqnarray*}
and for all role names $r$, we define
$$
r^{\Jmc_{j}}= \{ (w,wrv) \mid w,wrv \in \Delta^{\Jmc_{j}},v\in (\Delta^{\Imc_{j}} \cup \Delta^{\Imc_{1}}
\times \Delta^{\Imc_{2}})\}
$$
It is straightforward to prove that 
\begin{eqnarray*}
S_{1} & = & \{(e,w(e,f) \mid e\in \Delta^{\Imc_{1}}, w(e,f)\in \Delta^{\Jmc_{1}}\} \cup \\
& & \{(e,we) \mid e \in \Delta^{\Imc_{1}},we \in \Delta^{\Jmc_{1}}\}
\end{eqnarray*}
is a $\Sigma$-bisimulation between $(\Imc_{1},\rho^{\Imc_{1}})$ and $(\Jmc_{1},\rho^{\Jmc_{1}})$.
A $\Sigma$-bisimulation $S_{2}$ between $(\Imc_{2},\rho^{\Imc_{2}})$ and $(\Jmc_{2},\rho^{\Jmc_{2}})$
can be constructed in the same way. Clearly, $\Jmc_{1}^{\leq m} = \Jmc_{2}^{\leq m}$. 
\end{proof}

\medskip
\noindent
{\bf Theorem~\ref{thm1}.}
Let $\Tmc$ be a TBox, $\Sigma\subseteq {\sf sig}(\Tmc)$ a signature,
and $m\geq 0$. Then $\Tmc_{\Sigma,m}$ is not a uniform $\Sigma$-interpolant of $\Tmc$
iff

\medskip

\noindent
($\ast_{m}$) there exist two $\Sigma$-tree interpretations, 
$\mathcal{I}_{1}$ and $\mathcal{I}_{2}$, of finite outdegree such that
\begin{enumerate}
\item $\mathcal{I}_{1}^{\leq m} \;=\;  \mathcal{I}_{2}^{\leq m}$;
\item $(\mathcal{I}_{1},\rho^{\mathcal{I}_{1}})\models  \exists^\sim_{\overline{\Sigma}}. \Tmc$;
\item $(\mathcal{I}_{2},\rho^{\mathcal{I}_{2}}) \not\models \exists^\sim_{\overline{\Sigma}}. \Tmc$;
\item For all sons $d$ of $\rho^{\mathcal{I}_{2}}$:
$(\mathcal{I}_{2},d) \models  \exists^\sim_{\Sigma}. \Tmc$.
\end{enumerate}

\begin{proof}
Assume first that $\Tmc_{\Sigma,m}$ is not a uniform $\Sigma$-interpolant of $\Tmc$.
We show ($\ast_{m}$). There exists
$m'\geq m$ such that $\Tmc_{\Sigma,m'}\not\models \Tmc_{\Sigma,m'+1}$.
There exists a $\Sigma$-tree interpretation $\mathcal{I}$ of finite outdegree 
such that $\mathcal{I}\models \Tmc_{\Sigma,m'}$ and $\rho^{\mathcal{I}}\not\in C^{\mathcal{I}}$ for some 
$C$ with $\top \sqsubseteq C \in \Tmc_{\Sigma,m'+1}$.

As $\mathcal{I}\models \Tmc_{\Sigma,m'}$, the concept
$$
D = \bigsqcap_{E\in \mathcal{C}_{f}^{m'}(\Sigma),\rho^{\mathcal{I}}\in E^\Imc}E
$$
is satisfiable w.r.t.~$\Tmc$. There exists a $\Sigma$-tree interpretation $\mathcal{J}$ of finite outdegree 
that is a model of $\Tmc$ such that $(\Imc,\rho^{\Imc}) \equiv^{m'}_{\Sigma} (\Jmc,\rho^{\Jmc})$.
By Lemma~\ref{lem1}, we have 
$(\mathcal{I},\rho^{\mathcal{I}}) \sim_{\Sigma}^{m'} (\mathcal{J},\rho^{\Jmc})$.
By Lemma~\ref{lem2}, and closure under composition of (m)-bisimulations,
we can assume that $\Jmc$ is a $\Sigma$-tree interpretation of finite outdegree with  
\begin{itemize}
\item $\mathcal{I}^{\leq m'} \;=\;  \mathcal{J}^{\leq m'}$.
\item $(\mathcal{J},\rho^{\Jmc}) \models  \exists^\sim_{\Sigma}. \Tmc$.
\item $\mathcal{I}\models \Tmc_{\Sigma,m'}$.
\item $\rho^{\Imc}\not\in C^{\mathcal{I}}$ for some $(\top \sqsubseteq C)\in \Tmc_{\Sigma,m'+1}$.
\end{itemize}
For every son $d$ of $\rho^{\Imc}$, as $\mathcal{I}(d)$ is a model of $\Tmc_{\Sigma,m'}$ we can argue 
as above and find a $\Sigma$-tree interpretation $\mathcal{K}_{d}$ of finite outdegree such that
$$
(\mathcal{I}(d),d) \sim_{\Sigma}^{m'} (\mathcal{K}_{d},\rho^{\mathcal{K}_{d}}) 
\models  \exists^\sim_{\Sigma}. \Tmc.
$$
We have
$$
(\mathcal{J}(d),d) \sim_{m'-1}^{\Sigma} (\mathcal{K}_{d},\rho^{\mathcal{K}_{d}}).
$$
Thus, by Lemma~\ref{lem2}, we find $\Sigma$-tree interpretations
$\Jmc_{d}$ and $\mathcal{M}_{d}$ of finite outdegree such that
\begin{itemize}
\item $(\mathcal{J}_{d},\rho^{\Jmc_{d}}) \sim_{\Sigma} (\mathcal{J}(d),d)$;
\item $(\mathcal{K}_{d},\rho^{\Kmc_{d}}) \sim_{\Sigma} (\mathcal{M}_{d},\rho^{\mathcal{M}_{d}})$;
\item $\Jmc_{d}^{\leq m'-1} \;= \;  \mathcal{M}_{d}^{\leq m'-1}$.
\end{itemize}
Now define $\mathcal{I}_{1}$ by replacing, for every son $d$ of $\rho^{\Jmc}$,
$\mathcal{J}(d)$ by $\mathcal{J}_{d}$ in $\Jmc$. Define $\mathcal{I}_{2}$ by replacing, for every
son $d$ of $\rho^{\Imc}$, $\mathcal{I}(d)$ by $\mathcal{M}_{d}$ in $\Imc$. 
It is readily checked that $\mathcal{I}_{1}$ and $\mathcal{I}_{2}$ are
as required:
\begin{itemize}
\item $\mathcal{I}_{1}^{\leq m} \;=\;  \mathcal{I}_{2}^{\leq m}$: since
$m' \geq m$ it is sufficient to show
$\mathcal{I}_{1}^{\leq m'} \;=\;  \mathcal{I}_{2}^{\leq m'}$.
But since $\mathcal{I}^{\leq 1} \;=\;  \mathcal{J}^{\leq 1}$ this follows 
from $\Jmc_{d}^{\leq m'-1} \;= \;  \mathcal{M}_{d}^{\leq m'-1}$ for every son
$d$ of $\rho^{\Imc}$.
\item $(\mathcal{I}_{1},\rho^{\mathcal{I}_{1}}) \models  \exists^\sim_{\Sigma}. \Tmc$
follows from $(\Jmc,\rho^{\Jmc})\models  \exists^\sim_{\overline{\Sigma}}. \Tmc$ and
$(\Imc_{1},\rho^{\Imc_{1}})\sim_{\Sigma} (\Jmc,\rho^{\Jmc})$. 
\item $(\mathcal{I}_{2},\rho^{\mathcal{I}_{2}}) \not\models  \exists^\sim_{\overline{\Sigma}}. \Tmc$
follows if $\rho^{\Imc_{2}}\not\in C^{\Imc_{2}}$ for some $\Sigma$-concept $C$ such that 
$\Tmc \models \top \sqsubseteq C$. By construction,
there exists $C\in \mathcal{C}_{f}^{m'+1}(\Sigma)$ such that
$\rho^{\Imc}\not\in C^{\mathcal{I}}$ and $(\top \sqsubseteq C)\in \Tmc_{\Sigma,m'+1}$.
Thus, by Lemma~\ref{lem1}, it is sufficient to show
$(\Imc_{2},\rho^{\Imc_{2}}) \sim_{\Sigma}^{m'+1} (\Imc,\rho^{\Imc})$.
This follows if $(\Imc(d),d) \sim_{\Sigma}^{m'} (\mathcal{M}_{d},\rho^{\mathcal{M}_{d}})$,
for every son $d$ of $\rho^{\Imc}$.
But this follows from
$$
(\mathcal{I}(d),d) \sim_{\Sigma}^{m'} (\mathcal{K}_{d},\rho^{\mathcal{K}_{d}})
\sim_{\Sigma} (\mathcal{M}_{d},\rho^{\mathcal{M}_{d}}).
$$
\item For all sons $d$ of $\rho^{\mathcal{I}_{2}}$:
$(\mathcal{I}_{2},d) \models  \exists^\sim_{\Sigma}. \Tmc$. This follows from
$$
(\Imc_{2}(d),d) \;= \;(\mathcal{M}_{d},\rho^{\mathcal{M}_{d}}) \sim_{\Sigma}
(\mathcal{K}_{d},\rho^{\Kmc_{d}})\models  \exists^\sim_{\Sigma}. \Tmc.
$$
\end{itemize}
Now assume that $\mathcal{I}_{1}$ and $\mathcal{I}_{2}$ satisfy ($\ast_{m}$).
Then $\mathcal{I}_{2}$ is a model of $\Tmc_{\Sigma,m}$. For assume this is not the case.
Then $d\not\in C^{\Imc_{2}}$ for some
$C$ with $\top \sqsubseteq C \in \Tmc_{\Sigma,m}$. If $d=\rho^{\Imc_{2}}$, then $d\not\in C^{\Imc_{1}}$
because of Point~1. This contradicts $\Imc_{1} \models  \exists^\sim_{\Sigma}. \Tmc$.
If $d\not=\rho^{\Imc_{2}}$, then $d\in \mathcal{I}(d')$ for some son $d'$ of $\rho^{\Imc_{2}}$.
Hence $d\not\in C^{\Imc(d')}$, which contradicts Point~4.  

Now assume that $\Tmc_{\Sigma,m}$ is a uniform $\Sigma$-interpolant of $\mathcal{T}$.
As $\mathcal{I}_{2}$ is a model of $\Tmc_{\Sigma,m}$ and has finite outdegree,
we obtain from the modified version of Theorem~\ref{bisimuniform} proved above
that $(\mathcal{I}_{2},\rho^{\Imc_{2}}) \models  \exists^\sim_{\overline{\Sigma}}. \Tmc$, which contradicts
Point~3.
\end{proof}

We show Example~\ref{ex1} (iv): for $\Tmc_4$ consisting of
\begin{enumerate}
\item $A \sqsubseteq \exists r.B$;
\item $A_0 \sqsubseteq \exists r.(A_1 \sqcap B)$;
\item $E \equiv A_1 \sqcap B \sqcap \exists r.(A_2 \sqcap B)\}$;
\end{enumerate}
and $\Sigma_{4}=\{A,r,A_{0},A_{1},E\}$, there is no uniform $\Sigma_{4}$-interpolant of 
$\Tmc_4$. 

\begin{proof}
It is sufficient to show $(\ast_{m})$ for all $m>0$.
Let $\Imc_{1} = (\{0,\ldots,m+1,(a,2),\ldots,(a,m+1)\},\cdot^{\Imc_{1}})$, where
\begin{eqnarray*}
A^{\Imc_{1}} & = &\{1,\ldots,m\}\\
E^{\Imc_{1}} & =  &\emptyset\\
r^{\Imc_{1}} &=  &  \{(i,i+1) \mid 0\leq i \leq m\}\cup\\
           &   &  \{(i,(a,i+1))\mid 1\leq i \leq m\}\\
A_{1}^{\Imc_{1}} & = & \{1,\ldots,m+1\}\\
A_{2}^{\Imc_{1}} & = & \{(a,2),\ldots,(a,m+1)\}\\
A_{0}^{\Imc_{1}} & = & \{0\}
\end{eqnarray*}
Then $(\Imc_{1},0)\models \exists^{\sim}_{\overline{\Sigma}_{4}}.\Tmc_{4}$
because the expansion of $\Imc_{1}$ by 
$B^{\Imc_{1}}=\{1,\ldots,m+1\}$ is a model of $\Tmc_{4}$.

Define $\Imc_{2}$ as the restriction of $\Imc_{1}$ to $\Delta^{\Imc_{1}}\setminus\{m+1\}$.
By definition, $\Imc_{1}^{\leq m}\;=\;\Imc_{2}^{\leq m}$, 

\medskip

\noindent
Claim 1. $(\Imc_{2},0)\not\models \exists^{\sim}_{\overline{\Sigma}_{4}}.\Tmc_{4}$.

\medskip

\noindent Assume $(\Imc_{2},0)\models \exists^{\sim}_{\overline{\Sigma}_{4}}\Tmc_{4}$.
Take $(\Imc_{2},0)\sim_{\Sigma} (\Jmc,0')$ with $\Jmc$ a model of $\Tmc_{4}$. Let $S$ be the
$\Sigma$-bisimulation with $(0,0')\in S$. By inclusion (2.) there exists $1'$ with
$(1,1')\in S$ such that $1'\in (\neg E \sqcap A \sqcap A_{1}\sqcap B)^{\Jmc}$. By inclusion (1.) there
exists an $r$-successor $2'$ of $1'$ that is in $B^{\Jmc}$. We have $(2,2')\in S$ or $((a,2),2')\in S$.
But $((a,2),2')\not\in S$ because otherwise $2'\in (A_{2}\sqcap B)^{\Jmc}$ which, since $1'\in (A_{1}\sqcap B)^{\Jmc}$,
would imply, by inclusion (3.), that $1'\in E^{\Jmc}$, a contradiction. Thus, $(2,2')\in S$ and so
$2'\in (\neg E \sqcap A \sqcap A_{1}\sqcap B)^{\Jmc}$. One can now show in same way by induction
that there is a $m'$ with $(m,m')\in S$ such that $m'\in (\neg E \sqcap A \sqcap A_{1}\sqcap B)^{\Jmc}$.
All $r$-successors of $m'$ are in $A_{2}^{\Jmc}$ since all $r$-successors of $m$ are in $A_{2}^{\Jmc}$.
By inclusions (2.) and (3.) and since $m'\not\in E^{\Jmc}$ this leads to a contradiction.
  
\medskip

As $1$ is the only $r$-successor of $0$ in $\Imc_{2}$, it remains to show that
$(\Imc_{2},1)\models \exists^{\sim}_{\overline{\Sigma}_{4}}.\Tmc_{4}$. Let $\Imc_{2}'$
be the restriction of $\Imc_{2}$ to $\Delta^{\Imc_{2}}\setminus \{0\}$.
Then $(\Imc_{2},1)\models \exists^{\sim}_{\overline{\Sigma}_{4}}.\Tmc_{4}$ follows from
the observation that the expansion of $\Imc_{2}'$ by setting $B^{\Imc_{2}'}=\{(a,2),\ldots,(a,m+1)\}$
is a model of $\Tmc_{4}$.
\end{proof}
In a tree interpretation $\Imc$, we set ${\sf dist}(\rho^{\Imc},d)=k$ and say that the \emph{depth of $d$ in $\Imc$}
is $k$ iff $d$ can be reached from $\rho^{\Imc}$ in exactly $k$ steps.

\medskip

\noindent
{\bf Example~\ref{ex:forget}}
  Let $\Tmc$ be a TBox and $\Sigma$ a
  signature such that $\mn{sig}(\Tmc) \setminus \Sigma$ consists of
  stratified concept names only, i.e., we want to \emph{forget} a set of
  stratified concept names. Then the existence of a uniform
  $\Sigma$-interpolant of \Tmc is guaranteed; moreover,
  $\Tmc_{\Sigma,m}$ is such an interpolant, where $m=\max\{ {\sf
    rd}(C) \mid C \in \mn{conc}(\Tmc) \}$.  

\medskip
\begin{proof}
Assume $\Imc_{1},\Imc_{2}$ satisfy $(\ast_{m})$.
There exists a tree-interpretation $\Jmc_{1}$ that is a model of $\Tmc$ such that
$(\Imc_{1},\rho^{\Imc_{1}})\sim_{\Sigma} (\Jmc_{1},\rho^{\Jmc_{1}})$,
and for every $r$-successor $d_{r}$ of $\rho^{\Imc_{2}}$ there exists a tree-interpretation 
$\Jmc_{d_{r}}$ that is a model of $\Tmc$ with $(\Imc_{2},d_{r})\sim_{\Sigma} (\Jmc_{d_{r}},\rho_{d_{r}})$. We may assume that
$\Imc_{1}$ is the $\Sigma$-reduct of $\Jmc_{1}$ and that every $\Imc_{2}(d_{r})$ coincides 
with the $\Sigma$-reduct of $\Jmc_{d_{r}}$. Now expand $\Imc_{2}$ to an interpretation $\Jmc_{2}$ as 
follows: for every $B\in {\sf sig}(\Tmc)\setminus\Sigma$ of level $k\leq m$ set
\begin{eqnarray*}
B^{\Jmc_{2}} & = & \{ d\mid d\in B^{\Jmc_{1}} \wedge {\sf dist}(\rho^{\Imc_{2}},d)=k\}\cup\\
          &  &  \{ d\mid d\in B^{\Jmc_{d_{r}}} \wedge {\sf dist}(\rho^{\Imc_{2}},d)\not=k\}.
\end{eqnarray*}
We show that $\Jmc_{2}$ is a model of $\Tmc$; and have derived a contradiction as
$\Imc_{2}\not\models \exists^{\sim}_{\overline{\Sigma}}.\Tmc$.
Let $C\sqsubseteq D\in \Tmc$ and assume that $d\in C^{\Jmc_{2}}\setminus D^{\Jmc_{2}}$.
Let ${\sf dist}(\rho^{\Jmc_{2}},d)=l$. If $l=0$, then $d\in X^{\Jmc_{2}}$ iff $d\in X^{\Imc_{1}}$ for
all concepts $X$ in $\Tmc$, by the definition of the expansion. Thus, $d\in C^{\Imc_{1}}\setminus D^{\Imc_{1}}$
which contradicts that $\Imc_{1}$ is a model of $\Tmc$.
If $l>0$, then $d$ is in the domain of some $\Imc_{2}(d_{r})$. Then $d\in X^{\Jmc_{2}}$ iff $d\in X^{\Jmc_{d_{r}}}$ for
all concepts $X$ in $\Tmc$, by the definition of the expansion. Thus, $d\in C^{\Jmc_{d_{r}}}\setminus D^{\Jmc_{d_{r}}}$
which contradicts that $\Jmc_{d_{r}}$ is a model of $\Tmc$. 
\end{proof}

%
%
%

Fix a TBox $\Tmc$ and $\Sigma\subseteq {\sf sig}(\Tmc)$. Set $(\Imc_{1},d_{1})\sim_{e}(\Imc_{2},d_{2})$ iff 
${\sf Ext}^{\Imc_{1}}(d_{1})={\sf Ext}^{\Imc_{2}}(d_{2})$.
Note that the number of $\sim_{e}$-equivalence classes is bounded 
by $M_{\Tmc}$.

\begin{lemma}\label{lem3}
Let $\Imc$ be a tree interpretation and $d\in \Delta^{\mathcal{I}}$.
Assume $(\mathcal{I}(d),d) \sim^{e} (\mathcal{J},\rho^{\Jmc})$ for a tree interpretation
$\mathcal{J}$.
Replace $\mathcal{I}(d)$ by $\mathcal{J}$ in $\mathcal{I}$ and denote the resulting 
tree interpretation by $\Kmc$. Then $\mathcal{I}\models  \exists^\sim_{\Sigma}. \Tmc$ if, and only if,
$\mathcal{K} \models  \exists^\sim_{\Sigma}. \Tmc$.
\end{lemma}
\begin{proof}
Let $\Imc$, $d$, $\Jmc$, and $\Kmc$ be as in the formulation of Lemma~\ref{lem3}.
Assume $\mathcal{I} \models  \exists^\sim_{\Sigma}. \Tmc$. There exists a tree-interpretation
$\Imc'$ that is a model of $\Tmc$
such that $(\mathcal{I},\rho^{\Imc}) \sim_{\Sigma} (\mathcal{I}',\rho^{\Imc'})$. 
We may assume that there is a $\Sigma$-bisimulation $S$ between
$(\mathcal{I},\rho^{\Imc})$ and $(\mathcal{I}',\rho^{\Imc'})$ such that $S^{-}$ is an
injective relation and such that $(e,e')\in S$ implies that $e$ is reached from $\rho^{\Imc}$
along the same path as $e'$ from $\rho^{\Imc'}$. Let $S(d)= \{ d' \mid (d,d')\in S\}$.
Consider, for every $d'\in S(d)$, a tree-interpretation $\mathcal{K}_{d'}$ satisfying $\Tmc$ such that
\begin{itemize}
\item $\mn{tp}^{\mathcal{I}',\mathcal{T}}(d') = \mn{tp}^{\mathcal{K}_{d'},\mathcal{T}}(\rho^{\mathcal{K}_{d'}})$;
\item $(\mathcal{J},\rho^{\Jmc}) \sim_{\Sigma} (\mathcal{K}_{d'},\rho^{\Kmc_{d'}})$.
\end{itemize}
Such interpretations $\Kmc_{d'}$ exist by the definition of the equivalence relation
$\sim^{e}$. Now replace, in $\mathcal{I}'$ and for all $d'\in S(d)$, 
the tree interpretation $\mathcal{I}'(d')$ by $\mathcal{K}_{d'}$, and denote the resulting 
interpretation by $\Kmc'$. $\Kmc'$ is a model of $\Tmc$ since
$\mn{tp}^{\mathcal{I}',\mathcal{T}}(d') = \mn{tp}^{\mathcal{K}_{d'},\mathcal{T}}(\rho^{\mathcal{K}_{d}})$,
$\Imc'$ is a model of $\Tmc$ and all $\Kmc_{d'}$ are models of $\Tmc$.
It remains to show that $(\mathcal{K},\rho^{\Kmc}) \sim_{\Sigma} (\Kmc',\rho^{\Kmc'})$.
Take for every $d'\in S(d)$ a $\Sigma$-bisimulation $S_{d'}$ between
$(\mathcal{J},\rho^{\Jmc})$ and $(\mathcal{K}_{d'},\rho^{\Kmc_{d'}})$.
Let $S'$ be the restriction of $S$ to 
$$
(\Delta^{\mathcal{I}}\setminus\Delta^{\mathcal{I}(d)}) \times (\Delta^{\mathcal{I}'}
\setminus(\bigcup_{d'\in S(d)}\Delta^{\mathcal{I}'(d')}))
$$
It is not difficult to show that $S' \cup \bigcup_{d'\in S(d)}S_{d'}$ is the required 
$\Sigma$-bisimulation between $(\Kmc,\rho^{\Jmc})$ and $(\Kmc',\rho^{\Kmc})$.
\end{proof}

\medskip
\noindent
{\bf Theorem~\ref{fixm}.}
  Let $\Tmc$ be a TBox and $\Sigma\subseteq {\sf sig}(\Tmc)$. Then
  there does not exist a uniform $\Sigma$-interpolant of $\Tmc$
  iff $(\ast_{M_{\Tmc}^{2}+1})$ from Theorem~\ref{thm1} holds, where
  $M_{\Tmc}:=2^{2^{|\Tmc|}}$.

\medskip

\begin{proof}
By Theorem~\ref{thm1}, it is sufficient to prove that $(\ast_{M_{\Tmc}^{2}+1})$ implies $(\ast_{m})$ for all
$m\geq M_{\Tmc}^{2}+1$.  
%
%
Take $\Sigma$-tree interpretations $\mathcal{I}_{1}$ and $\mathcal{I}_{2}$ satisfying 
$(\ast_{m})$ of Theorem~\ref{thm1} for some $m \geq M_{\Tmc}^{2}+1$. 
We show that there exist $\Sigma$-tree interpretations 
$\mathcal{J}_{1}$ and $\mathcal{J}_{2}$ satisfying $(\ast_{m+1})$ of Theorem~\ref{thm1}.
The implication then follows by induction.

Let $D$ be the set of $d\in \Delta^{\Imc_{1}}$ with ${\sf dist}(\rho^{\Imc_{1}},d)=m$
such that $\Imc_{1}(d)^{\leq 1} \;\not=\; \Imc_{2}(d)^{\leq 1}$ (i.e., the
restrictions of $\Imc_{1}$ to
$\{ d' \mid d' \mbox{ son of $d$ in $\Imc_{1}$}\}$ and $\Imc_{2}$ to
$\{ d' \mid d' \mbox{ son of $d$ in $\Imc_{2}$} \}$ do not
coincide). If $D =\emptyset$, then $\Imc_{1}^{\leq m+1} \; = \; \Imc_{2}^{\leq m+1}$
and the claim is proved. Otherwise choose $f\in D$ and consider the path 
$$
\rho^{\Imc_{1}}= d_{0} r_{0} d_{1} \cdots r_{m-1} d_{m}= f
$$
with $(d_{i},d_{i+1}) \in r_{i}^{\Imc_{1}}$ for all $i<m$. 
As $m \geq M_{\Tmc}^{2}+1$, there exists $0<i<j\leq m$ such that both, 
$$
(\mathcal{I}_{1},d_{i}) \sim^{e} (\mathcal{I}_{1}, d_{j}),
\quad
(\mathcal{I}_{2},d_{i}) \sim^{e} (\mathcal{I}_{2}, d_{j}).
$$
Replace $\mathcal{I}_{1}(d_{j})$ by 
$\mathcal{I}_{1}(d_{i})$ in $\mathcal{I}_{1}$ 
and denote the resulting interpretation by $\mathcal{K}_{1}$. 
Similarly, replace $\mathcal{I}_{2}(d_{j})$ by 
$\mathcal{I}_{2}(d_{i})$ in $\mathcal{I}_{2}$ and denote the resulting interpretation
$\mathcal{K}_{2}$.
By Lemma~\ref{lem3}, $\mathcal{K}_{1}$ and $\mathcal{K}_{2}$ still have Properties (1)-(4).
Moreover, the set $D'$ of all $d\in \Delta^{\Imc_{1}'}$ with ${\sf dist}(\rho^{\Imc_{1}'},d)=m$ 
such that $\Kmc_{1}(d)^{\leq 1} \;\not= \Kmc_{2}(d)^{\leq 1}$ 
is a subset of $D$ not containing $f$. 
Thus, we can proceed with $D'$ in the same way as above until the set is empty. Denote
the resulting interpretations by $\mathcal{J}_{1}$ and $\mathcal{J}_{2}$,
respectively. They still have Properties (1)--(4), but now for some $m'>m$. 
\end{proof}

\section{Proofs for Section~\ref{sect:automatastuff}}



We start with establishing some basic results about APTAs. To
formulate and prove these, we make some technicalities more formal
than in the main paper. Recall than a run is a pair $(T,\ell)$ with
$T$ a tree. Let us make precise what exactly we mean by `tree' here. A
\emph{tree} is a non-empty (finite or infinite) prefix-closed subset
$T \subseteq S^*$, for some set $S$. If $d \in T$ and $d \cdot c \in
T$ with $d \in S^*$ and $c \in S$, then the node $d \cdot c$ is a
\emph{son} of the node $d$ in $T$.  A node $d \in T$ that has no sons
is a \emph{leaf}.  We measure the size of an APTA primarily in the
number of states. To define a more fine-grained measure, we use
$||\Amc||$ to denote the \emph{size} of \Amc, i.e.,
$|Q|+{|\Sigma_N|}+|\Sigma_E|$. Note that the size of (the
representation of) all other components of the automaton is bounded
polynomially in $||\Amc||$. In particular, we can w.l.o.g.\ assume
that the values in $\Omega$ are bounded by $2|Q|$.
\begin{lemma}
\label{lem:compl}
Let $\Amc_i=(Q_i,\Sigma_{N},\Sigma_{E},q_{0,i},\delta_i,\Omega_i)$ be
APTAs, $i \in \{1,2\}$ with $Q_1 \cap Q_2 = \emptyset$. Then there is an APTA
\begin{enumerate}

\item 
  $\Amc'=(Q_1,\Sigma_{N},\Sigma_{E},q_{0,1},\delta',\Omega')$ such that
  $L(\Amc')=\overline{L(\Amc_1)}$;

\item 
  $\Amc''=(Q_1 \uplus Q_2 \uplus \{q_0\},\Sigma_{N},\Sigma_{E},q_{0},\delta'',\Omega'')$ such that
  $L(\Amc'')=L(\Amc_1) \cap L(\Amc_2)$. 

\end{enumerate}
Moreover, $\Amc'$ and $\Amc''$ can be constructed in time $p(||\Amc||)$, $p$ a polynomial.
\end{lemma}
\begin{proof} (sketch) The construction of $\Amc'$ is based on the
  standard dualization construction first given in
  \cite{Muller-Schupp-87}, i.e., $\delta'$ is obtained from $\delta$
  by swapping \mn{true} and \mn{false}, $A$ and $\neg A$, $\wedge$ and
  $\vee$, and diamonds and boxes, and setting $\Omega'(q)=\Omega(q)+1$
  for all $q \in Q$. The construction of $\Amc''$ is standard as well:
  add a fresh initial state $q_0$ with $\delta(q_0)=q_{0,1} \wedge
  q_{0,1}$, and define $\delta''$ and $\Omega''$ as the fusion of the
  respective components of $\Amc_1$ and $\Amc_2$ (e.g.,
  $\delta''(q)=\delta_1(q)$ for all $q \in Q_1$ and
  $\delta''(q)=\delta_2(q)$ for all $q \in Q_2$).
\end{proof}
The following lemma shows that a language accepted by an APTA is
closed under bisimulation. It implies that whenever for an APTA \Amc
we have $L(\Amc) \neq \emptyset$, then there is a pointed tree
interpretation $(\Imc,d)$ with $(\Imc,d) \in L(\Amc)$.
As a notational convention, whenever $x$ is a node in a $Q \times
\Delta^\Imc$-labelled tree and $\ell(x)=(q,d)$, then we use
$\ell_1(x)$ to denote $q$ and $\ell_2(x)$ to denote $d$. 
\begin{lemma}
\label{lem:aptabisimclosed}
  Let $\Amc=(Q,\Sigma_{N},\Sigma_{E},q_{0},\delta,\Omega)$ be an APTA,
  $(\Imc,d) \in L(\Amc)$, and $(\Imc,d) \sim_{\Sigma_N \cup \Sigma_E} (\Jmc,e)$.
  Then $(\Jmc,e) \in L(\Amc)$.
\end{lemma}
\begin{proof}
  Let $(\Imc,d) \in L(\Amc)$, and $(\Imc,d) \sim_{\Sigma_N \cup
    \Sigma_E} (\Jmc,e)$. Moreover, let $(T,\ell)$ be an accepting run
  of \Amc on $(\Imc,d)$.  We inductively construct a $Q \times
  \Delta^{\mathcal{J}}$-labelled tree $(T',\ell')$, along with a map
  $\mu:T' \rightarrow T$ such that $\mu(y)=x$ implies
  $\ell_1(x)=\ell'_1(y)$ and $(\Imc,\ell_2(x)) \sim_{\Sigma_N \cup
    \Sigma_E} (\Jmc,\ell'_2(y))$:
  \begin{itemize}

  \item start with $T'=\{\varepsilon\}$, $\ell'(\varepsilon)=(q,e)$, and
    $\mu(\varepsilon)=\varepsilon$;

  \item if $y \in T'$ is a leaf, $\ell'_1(y)=q' \wedge q''$, and
    $\mu(y)=x$, then there are sons $x',x''$ of $x$ with
    $\ell(x')=(q',\ell_2(x))$ and $\ell(x'')=(q'',\ell_2(x))$; add
    fresh $y \cdot c'$ and $y \cdot c''$ to $T'$ and put $\ell'(y
    \cdot c')=(q',\ell'_2(y))$, $\ell'(y \cdot c'')=(q'',\ell'_2(y))$,
    $\mu(y \cdot c')=x'$, and $\mu(y \cdot c'')=x''$;

  \item if $y \in T'$ is a leaf, $\ell'_1(y)=q' \vee q''$, and
    $\mu(y)=x$, then there is a son $x'$ of $x$ with $\ell_1(x') \in
    \{ q',q'' \}$ and $\ell_2(x')=\ell_2(x)$; add a fresh $y \cdot c'$
    to $T'$ and put $\ell'(y \cdot c')=(\ell_1(x'),\ell'_2(y))$ and
    $\mu(y \cdot c')=x'$;

  \item if $y \in T'$ is a leaf, $\ell'_1(y)=\langle r \rangle q'$,
    and $\mu(y)=x$, then there is an $(\ell_2(x),d) \in r^\Imc$ and a
    son $x'$ of $x$ with $\ell(x')=(q',d)$; since $(\Imc,\ell_2(x))
    \sim_{\Sigma_N \cup \Sigma_E} (\Jmc,\ell'_2(y))$, there is an
    $(\ell'_2(y),d') \in r^\Jmc$ with $(\Imc,d) \sim_{\Sigma_N \cup
      \Sigma_E} (\Jmc,d')$; add a fresh $y \cdot c'$ to $T'$ and put
    $\ell'(y \cdot c')=(q',d')$ and $\mu(y \cdot c')=x'$;

  \item if $y \in T'$ is a leaf, $\ell'_1(y)=[ r ] q'$, and
    $\mu(y)=x$, then do the following for every $(\ell_2'(y),d') \in
    r^\Jmc$: since $(\Imc,\ell_2(x)) \sim_{\Sigma_N \cup \Sigma_E}
    (\Jmc,\ell'_2(y))$, there is an $(\ell_2(x),d) \in r^\Imc$ with
    $(\Imc,d) \sim_{\Sigma_N \cup \Sigma_E} (\Jmc,d')$, and thus also
    a son $x'$ of $x$ with $\ell(x')=(q',d)$; add a fresh $y \cdot c'$
    to $T'$ and put $\ell'(y \cdot c')=(q',d')$ and $\mu(y \cdot
    c')=x'$.

  \end{itemize}
  It can be verified that $(T',\ell')$ is an accepting run of \Amc on
  $(\Jmc,e)$.
\end{proof}
Finally, we fix the complexity of the emptiness problem of APTAs.
The following is proved in \cite{Wilke-Automata} using a reduction to 
parity games.
\begin{theorem}[Wilke]
\label{thm:emptiness}
Let $\Amc=(Q,\Sigma_{N},\Sigma_{E},q_{0},\delta,\Omega)$ be an
APTA.  Then the emptiness of $L(\Amc)$
can be decided in time $2^{p(||\Amc||)}$, $p$ a polynomial.
\end{theorem}
The following lemma establishes the correctness of the construction of
the automata $\Amc_{\Tmc_\Sigma}$ for Theorem~\ref{thm:bisimauto}.
\begin{lemma}
\label{lem:autocorrect}
$(\Imc,d_0) \in
  L(\Amc_{\Tmc,\Sigma})$ iff $(\Imc,d_0) \models \exists^\sim_{\overline{\Sigma}}
  . \Tmc$,
  for all pointed $\Sigma$-interpretations $(\Imc,d_0)$.
\end{lemma}
\begin{proof}
  Let $(\Imc,d_0)$ be a pointed $\Sigma$-interpretation.  By
  definition of $\Amc_{\Tmc,\Sigma}$, we have $(\Imc,d_0) \in
  L(\Amc_{\Tmc,\Sigma})$ iff there exists a $Q \times
  \Delta^{\mathcal{I}}$-labelled tree $(T,\ell)$ such that
\begin{enumerate}

\item $\ell(\varepsilon) = (q_{0},d_0)$; 

\item there exists a son $x$ of $\varepsilon$ and 
$t\in \mn{TP}(\Tmc)$ such that 
$\ell(x) = (t,d_0)$;

\item if $\ell(x) = (t,d)$, then 

\begin{itemize}

\item[(a)] $d\in A^{\mathcal{I}}$ for all $A \in t \cap \NC \cap \Sigma$;

\item[(b)] $d\not\in A^{\mathcal{I}}$ for all $A\in (\NC \cap \Sigma)\setminus t$;

\item[(c)] for all $r \in \Sigma$ and $(d,d')\in r^{\Imc}$,
there exist $t'\in \mn{TP}(\Tmc)$ such that $t \leadsto_{r} t'$ 
and a son $y$ of $x$ such that $\ell(y) = (t',d')$;

\item[(d)] for all $\exists r.C\in t$ with $r \in \Sigma$, there exists
  $(d,d')\in r^{\Imc}$ and $t'\in \mn{TP}(\Tmc)$ such that $C \in t'$
  and $t \leadsto_{r} t'$, and $\ell(y) = (t',d')$ for some son $y$ of $x$.

\end{itemize}

\end{enumerate}
Assume that $(\Imc,d_0) \models \exists^\sim_{\overline{\Sigma}} . \Tmc$ and let
$(\Jmc,e_0)$ be a pointed model of $\Tmc$ such that $(\Imc,d_0)
\sim_{\Sigma} (\Jmc,e_0)$. A \emph{path} is a sequence
$(d_1,e_1)\cdots(d_n,e_n)$, $n \geq 0$, with $d_1,\dots,d_n \in
\Delta^\Imc$ and $e_1,\dots,e_n \in \Delta^\Jmc$ such that
\begin{itemize}

\item[(i)] $d_1 = d_0$;

\item[(ii)] $(d_i,d_{i+1}) \in r^\Imc$ for some $r \in \Sigma \cap \NR$, for
  $1 \leq i < n$;

\item[(iii)] $(\Imc,d_i) \sim_\Sigma (\Jmc,e_i)$ for $1 \leq i \leq n$.

\end{itemize}
Define a $Q \times \Delta^{\mathcal{I}}$-labelled tree $(T,\ell)$ by
setting
\begin{itemize}

\item $T$ to the set of all paths;

\item $\ell(\varepsilon)=(q_0,d_0)$;

\item $\ell((d_1,e_1)\cdots(d_n,e_n))=(\mn{tp}^\Jmc(e_n),d_n)$ for all
  paths $(d_1,e_1)\cdots(d_n,e_n) \neq \varepsilon$.
  
\end{itemize}
One can now verify that $(T,\ell)$ satisfies Conditions~1 to~3,
thus $(\Imc,d_0) \in L(\Amc_{\Tmc,\Sigma})$. In fact, Conditions~1 and~2
are immediate and Conditions~3a and~3b are a consequence of the
definition of $\ell$ and (iii). As for Condition~3c, let
$\ell(x)=(t,d)$, $(d,d') \in r^\Imc$, and $r \in \Sigma$, and assume
$x= (d_1,e_1)\cdots(d_n,e_n)$. Then $d=d_n$, and $d_n \sim_\Sigma e_n$
and $(d,d') \in r^\Imc$ yield an $(e_n,e') \in r^\Jmc$ with $d'
\sim_\Sigma e'$. Thus $y:=(d_1,e_1)\cdots(d_n,e_n)(d',e')$ is a son of
$x$ and $t':=\mn{tp}^\Jmc(e')$ is as desired, i.e., $t
\rightsquigarrow_r t'$ and $\ell(y)=(t',d')$. Condition~3d can be
established similarly.

\smallskip

Conversely, assume that there is a $Q \times
\Delta^{\mathcal{I}}$-labelled tree $(T,\ell)$ that satisfies
Conditions~1 to~3. 
Define an interpretation
${\Jmc_0}$ by setting
\begin{itemize}

\item $\Delta^{{\Jmc_0}}= T\setminus \{x \in T \mid \ell_1(x)=q_{0} \}$;

\item $A^{{\Jmc_0}} = \{ x\in \Delta^{{\Jmc_0}} \mid A \in \ell_1(x) \}$;

\item $(x,y)\in r^{{\Jmc_0}}$ iff $y$ is a son of $x$, $\ell(x)=(t,d)$,
  $\ell(y)=(t',d')$, $t \leadsto_{r} t'$, and $(d,d')\in r^{\Imc}$.

\end{itemize}
The next step is to extend ${\Jmc_0}$ to also satisfy existential
restrictions $\exists r .C$ with $r \notin \Sigma$.  For each $x \in
\Delta^{\Jmc_0}$ and $\exists r . C \in \ell_1(x)$ with $r \notin
\Sigma$, fix a model $\Imc_{x,\exists r . C}$ of \Tmc that satisfies
$C$ and every $D$ with $\forall r . D \in \ell_1(x)$ at the root. Such
models exist since $\ell_1(x) \in \mn{TP}(\Tmc)$, thus it is realized
in some model of \Tmc. Let $\Imc_{x_1,\exists r_1 . C_1},
\dots,\Imc_{x_k,\exists r_k . C_k}$ be the chosen models and assume
w.l.o.g.\ that their domains are pairwise disjoint, and also disjoint
from $\Delta^{\Jmc_0}$. Now define a new interpretation \Jmc as
follows:
\begin{itemize}

\item $\Delta^\Jmc = \Delta^{\Jmc_0} \cup \displaystyle\bigcup_{1 \leq i \leq k}
  \Delta^{\Imc_{x_i,\exists r_i . C_i}}$;

\item $A^\Jmc = A^{\Jmc_0}  \cup \displaystyle\bigcup_{1 \leq i \leq k}
  A^{\Imc_{x_i,\exists r_i . C_i}}$;

\item $r^\Jmc = r^{\Jmc_0}  \cup \displaystyle\bigcup_{1 \leq i \leq k}
  r^{\Imc_{x_i,\exists r_i . C_i}} \cup \bigcup_{1 \leq i \leq k}
  (x_i,\rho^{\Imc_{x_i,\exists r_i . C_i}})$.

\end{itemize}
By Condition~2, there is a son $x_0$ of $\varepsilon$ in $T$ such that
$\ell(x_0)=(t,d_0)$. Using Condition~3, it can be verified that
$
\{(d,x) \in \Delta^{\Imc}\times \Delta^{\Jmc_0} \mid \ell_2(x)=d \}
$
is a $\Sigma$-bisimulation between $(\Imc,d_0)$ and $(\Jmc,x_0)$. It
thus remains to show that \Jmc is a model of \Tmc, which is an
immediate consequence of the following claim and the definition of
types for \Tmc.
\\[2mm]
{\bf Claim}. For all $C \in \mn{cl}(\Tmc)$:
\begin{enumerate}

\item[(i)] for all $x \in \Delta^{\Jmc_0}$ with $\ell(x)=(t,d)$, $C \in t$
  implies $x \in C^\Jmc$;

\item[(ii)] for $1 \leq i \leq k$ and all $x \in \Delta^\Imc_{x_i,\exists
    r_i . C_i}$,  $x \in C^{\Imc_{x_i,\exists r_i . C_i}}$ implies
  $x \in C^\Jmc$.

\end{enumerate}
The proof is by induction on the structure of $C$. We only do the case
$C=\exists r . D$ explicitly. For Point~(i), let $x \in \Delta^{\Jmc_0}$
with $\ell(x)=(t,d)$, and $\exists r . D \in t$. First assume $r \in
\Sigma$. Then Condition~3d yields a $(d,d') \in r^\Imc$ and $t' \in
\mn{TP}(\Tmc)$ such that $D \in t'$ and $t \rightsquigarrow_r t'$, and
a son $y$ of $x$ with $\ell(y)=(t',d')$. By definition of $\Jmc_0$,
$(x,y) \in r^\Jmc$. By IH, $D \in t'$ yields $y \in D^\Jmc$, thus $x
\in (\exists r . D)^\Jmc$. Now assume $r \notin \Sigma$. Then
$(x,\rho^{\Imc_{x,\exists r . D}}) \in r^\Jmc$. By IH and choice of
$\rho^{\Imc_{x,\exists r . D}}$, $\rho^{\Imc_{x,\exists r . D}} \in D^\Jmc$
and we are done. For Point~(ii), it suffices to apply IH and the semantics.
\end{proof}

\medskip
\noindent
{\bf Proof of Theorem~\ref{bisimuniform}.}
We have already proved the modified version of the theorem, where
``for all interpretations \Imc'' is replaced with ``for all
interpretations \Imc with finite outdegree''. The ``if'' direction of
the modified version immediately implies the one of the original
version. For the ``only if'' direction of the original version, assume
that $\Tmc_\Sigma$ is a uniform $\Sigma$-interpolant of \Tmc.
The direction ``$\Leftarrow$'' of ($*$) is proved exactly as in the
modified version. For ``$\Rightarrow$'', take an interpretation \Imc
with $\Imc \models \Tmc_\Sigma$ and assume to the contrary that there
is a $d \in \Delta^\Imc$ such that $(\Imc,d)$ is not
$\Sigma$-bisimilar to any pointed model of \Tmc. Let \Amc be the
complement of the automaton $\Amc_{\Tmc,\Sigma}$ of
Theorem~\ref{thm:bisimauto}. By Lemma~\ref{lem:autocorrect}, $(\Imc,d)
\in L(\Amc)$. We can w.l.o.g.\ assume that $\Imc$ is a tree
interpretation with root $d$ (if it is not, apply unravelling). By
considering an accepting run $(T,\ell)$ of \Amc on \Imc and removing
unnecessary subtrees synchronously from both \Imc and $(T,\ell)$, it
is easy to show that there is a $(\Jmc,d) \in L(\Amc)$ that is still a
model of $\Tmc_\Sigma$, but of finite outdegree. Since $\Jmc \in
L(\Amc)$, \Jmc is not $\Sigma$-bisimilar to any model of $\Tmc$. The
existence of such a \Jmc contradicts the modified
Theorem~\ref{bisimuniform}, which we already proved to hold.  \qed

\medskip
\noindent
Our next aim is to prove Theorem~\ref{prop:cetocons}, which or convenience
we state in expanded form here.
{\bf Theorem~\ref{prop:cetocons}.}
Let \Tmc be a TBox, $\Sigma \subseteq \mn{sig}(\Tmc)$ a signature, and
$m \geq 0$. Then there is an APTA
$\Amc_{\Tmc,\Sigma,m}=(Q,\Sigma_N,\Sigma_E,q_0, \delta, \Omega)$ such
that $L(\Amc) \neq \emptyset$ iff  there are $\Sigma$-tree
interpretations $(\Imc_1,d_1)$ and $(\Imc_2,d_2)$ such that
\begin{enumerate}

\item $\Imc_1^{\leq m}=\Imc_2^{\leq m}$;

\item $(\Imc_1,\rho^{\Imc_1}) \models \exists^\sim_{\overline{\Sigma}} . \Tmc$;

\item $(\Imc_2,\rho^{\Imc_2}) \not\models \exists^\sim_{\overline{\Sigma}} . \Tmc$;

\item for all successors $d$ of $\rho^{\Imc_2}$, we have $(\Imc_2,d) \models \exists^\sim_{\overline{\Sigma}} . \Tmc$.

\end{enumerate}
Moreover, $|Q| \in \Omc(2^{\Omc(n)}+\log^2 m)$ and $|\Sigma_N|,|\Sigma_E| \in \Omc(n+\log m)$,
where $n=|\Tmc|$.

\medskip
\noindent
\begin{proof}
%
%
  Let \Tmc be a TBox, $\Sigma \subseteq \mn{sig}(\Tmc)$ a signature,
  and $k \geq 0$.  We show how to constuct the APTA
  $\Amc_{\Tmc,\Sigma,k}$ stipulated in Theorem~\ref{prop:cetocons} as
  an intersection of four automata $\Amc_1,\dots,\Amc_4$. Let
  $k=\lceil\log(m+2)\rceil$. All of the automata $\Amc_i$ will use the
  alphabets
$$
  \begin{array}{rcl}
  \Sigma_N &=& ((\Sigma \cap \NC) \times \{ 1,2\}) \cup \{ c_1,\dots,c_k\} \\[1mm]
  \Sigma_E &=& (\Sigma \cap \NR) \times \{ 1,2,12\}
\end{array}
$$
We assume w.l.o.g.\ that $\Sigma_N \subseteq \NC$ and $\Sigma_E
\subseteq \NR$.  Intuitively, a $\Sigma_N \cup
\Sigma_E$-interpretation $\Imc$ represents two
$\Sigma$-interpretations $\Imc_1$ and $\Imc_2$ where for $i \in
\{1,2\}$, we set $\Delta^{\Imc_i}=\Delta^\Imc$, $A^{\Imc_i} = d \in
(A,i)^{\Imc}$ and $r^{\Imc_i}=(r,i)^{\Imc} \cup (r,12)^{\Imc}$. Thus,
edges indexed by ``12'' represent edges that are shared between the
two interpretations.  The additional concept names $c_1,\dots,c_k$ are
used to implement a counter that counts the depth of elements in
$\Imc$ up to $m+1$, and then stays at $m+1$.  As a by-product, the
counter ensures that every element has a uniquely defined
depth. Formally, a pointed $\Sigma_N \cup \Sigma_E$-interpretation
$(\Imc,d)$ is called \emph{$m$-well-counting} if for all $d_0,\dots,d_n
\in \Delta^\Imc$ with $d_0=d$ and $(d_i,d_{i+1}) \in \bigcup_{r \in
  \Sigma_E} r^\Imc$ for $0 \leq i < n$, the value encoded in binary by
the truth of the concept names $c_1,\dots,c_k$ at $d_n$ is $\min\{n,m+1\}$.

\smallskip

Automaton $\Amc_1$ ensures Condition~2 of Theorem~\ref{prop:cetocons}.
We construct it by starting with the APTA $\Amc_{\Tmc,\Sigma}$ from
Theorem~\ref{thm:bisimauto}, and then modifying it as follows to run
on track~1 of the combined interpretations:
    \begin{itemize}

    \item replace the node alphabet $\Sigma \cap \NR$ with $\Sigma_N$, 
      and the edge alphabet $\Sigma \cap \NC$ with $\Sigma_E$

    \item for all states $q$, $\delta(q)$ is obtained from $\delta(q)$
      by replacing $A$ with $(A,1)$, $\neg A$ with $\neg(A,1)$,
      $\langle r \rangle q$ with $\langle (r,1) \rangle q \vee \langle
      (r,12) \rangle q$ and every $[ r] q$ with $[ (r,1) ] q \wedge 
      [ (r,12) ] q$

    \end{itemize}
    Automaton $\Amc_2$ takes care of Condition~3 of
    Theorem~\ref{prop:cetocons}.  We again start with
    $\Amc_{\Tmc,\Sigma}$, first complement it according to
    Lemma~\ref{lem:compl} and then modify it as $\Amc_1$, but using
    track/index 2 instead of track/index 1. Automaton $\Amc_3$
    addresses Condition~4 of Theorem~\ref{prop:cetocons}.  We start
    once more with $\Amc_{\Tmc,\Sigma}$, modify it to run on track 2,
    add a new initial state $q_0'$, and put $\delta(q'_0)=\bigwedge_{r
      \in (\Sigma \cap \NR) \times \{2,12\}} [r] q_0$, with $q_0$ the
    original initial state, to start the run of the obtained automaton
    at every successor of the selected point instead of at the selected
    point itself. The following lemma states the
    central property of the APTAs constructed so far.
\\[2mm]
{\bf Claim~1}  Let $(\Imc,d)$ be a pointed $\Sigma_N \cup \Sigma_E$-interpretation. Then
  \begin{enumerate}
  \item 
  $(\Imc,d) \in L(\Amc_1)$ iff $(\Imc_1,d) \in L(\Amc_{\Tmc,\Sigma})$ iff $(\Imc_1,d) \models \exists^\sim_{\overline{\Sigma}} .\Tmc$
  \item 
  $(\Imc,d) \in L(\Amc_2)$ iff $(\Imc_2,d) \notin L(\Amc_{\Tmc,\Sigma})$ iff
 $(\Imc_2,d) \not\models \exists^\sim_{\overline{\Sigma}} .\Tmc$

  \item $(\Imc,d) \in L(\Amc_3)$ iff $(\Imc_2(d'),d') \in
  L(\Amc_{\Tmc,\Sigma})$ for all successors $d'$ of $d$ iff $(\Imc_2(d'),d')
  \models \exists^\sim_{\overline{\Sigma}} .\Tmc$ for all successors $d'$ of $d$.

  \end{enumerate}
The purpose of the final automaton $\Amc_4$ is to address Condition~1
of Theorem~\ref{prop:cetocons}. To achieve this, $\Amc_4$ also enforces that 
accepted interpretations are $m$-well-counting.
$$
\begin{array}{r@{\;}c@{\;}l}
Q&=&\{q_0, q_1, q_2\} \\[1mm]
\delta(q_0)&=&\neg c_1 \wedge \cdots \wedge \neg c_k \wedge q_1 \wedge q_2
\\[1mm] 
\delta(q_1)&=&((c = k+1) \vee [r] (c{+}{+})) \, \wedge \\[1mm]
&& ((c<k+1) \vee [r](c{=}{=}))
\wedge \displaystyle\bigwedge_{r \in \Sigma_E} [r]q_1
\\[5mm] 
\delta(q_2)&=&\displaystyle \bigwedge_{A \in \Sigma \cap \NC} ((A,1) \wedge (A,2)) \vee (\neg(A,1)
  \wedge \neg(A,2))
  \, \wedge \\[5mm]
&& \displaystyle \bigwedge_{r \in \Sigma \cap \NR} [(r,1)] \mn{false} \wedge [(r,2)] \mn{false} \\[5mm]
&& \displaystyle \wedge \; ((c=k+1) \vee \bigwedge_{r \in \Sigma \cap \NR} [(r,12)]q_2 )
\end{array}
$$
where $(c < k+1)$ and $(c=k+1)$ are the obvious Boolean NNF formulas
expressing that the counter $c_1,\dots,c_k$ is smaller and equal to
$k+1$, respectively, $[r] (c{=}{=})$ is a formula expressing that the
counter value does not change when travelling to $r$-successors, and
$[r] (c{+}{+})$ expresses that the counter is incremented when
travelling to $r$-successors. It is standard to work out the details
of these formulas. 
\\[2mm]
{\bf Claim~2.}  Let $(\Imc,d)$ be a pointed $\Sigma_N \cup
\Sigma_E$-interpretation. Then $(\Imc,d) \in L(\Amc_4)$ iff $(\Imc,d)$
is $k$-well-counting and $(\Imc_1,d)^{\leq k}=(\Imc_2,d)^{\leq
  k}$.\footnote{The definition of $(\Imc,d)^{\leq k}$ generalizes from
  tree interpretations to $k$-well-counting interpretations in the
  obvious way.}
\\[2mm]
By Claims~1 and~2 and Lemma~\ref{lem:aptabisimclosed}, the
intersection of $\Amc_1,\dots,\Amc_4$ satisfies Conditions~1-4 of
Theorem~\ref{prop:cetocons}. The size bounds stated in
Theorem~\ref{prop:cetocons} are also satisfied (note the additional
states implicit in $\Amc_4$).
\end{proof}

\medskip
\noindent
{\bf Theorem~\ref{thm:ceupper}.}
  Given TBoxes \Tmc and $\Tmc'$, it can be decided in time
  $2^{p(|\Tmc| \cdot 2^{|\Tmc'|})}$ whether $\Tmc \cup \Tmc'$ is a
  conservative extension of $\Tmc$, for some polynomial $p()$.

\medskip
\noindent
\begin{proof}(sketch) Let $\Tmc$ and $\Tmc'$ be TBoxes and
  $\Sigma=\mn{sig}(\Tmc)$. We construct an APTA \Amc such that
  $L(\Amc)=\emptyset$ iff there are $\Sigma$-tree interpretations
  $\Imc$ and $\Imc'$ such that $\Imc \models \Tmc$ and $\Imc
  \not\models \exists^\sim_{\overline{\Sigma}} . \Tmc'$. By Theorem~\ref{bisimCE}
  and a straightforward unravelling argument, it follows that $\Tmc
  \cup \Tmc'$ is a conservative extension of \Tmc iff
  $L(\Amc)=\emptyset$. We start with taking automata
  $\Amc_{\Tmc,\Sigma}$ and $\Amc_{\Tmc',\Sigma}$, where the latter is
  constructed according to Theorem~\ref{thm:bisimauto} and the former is
  defined as $(Q,\Sigma_N,\Sigma_E,q_0, \delta, \Omega)$, where $Q$ is
  the set of all subformulas of $C_\Tmc$,
$$
\begin{array}{r@{\;}c@{\;}l@{}r@{\;}c@{\;}l}
\Sigma_N &=& \Sigma \cap \NC &
\Sigma_E &=& \Sigma \cap \NR \\[1mm]
\delta(q_0)&=&\displaystyle C_\Tmc \wedge \bigwedge_{r \in \Sigma \cap \NR} [r] q_0  \\
\delta(A)&=&A &
\delta(\neg A)&=&\neg A\\[1mm]
\delta(C \sqcap D)&=&C \wedge D&
\delta(C \sqcup D)&=&C \vee D\\[1mm]
\delta(\exists r . C)&=&\langle r \rangle C&
\delta(\forall r . C)&=&[ r ] C\\[1mm]

\end{array}
$$
Set $\Omega(q)=0$ for all $q \in Q$. We can then use the constructions
from Lemma~\ref{lem:compl} to obtain a final automaton \Amc such
that $L(\Amc)=L(\Amc_{\Tmc,\Sigma}) \cap
\overline{L(\Amc_{\Tmc',\Sigma})}$. It can be verified that the automaton
is as required, and that Theorem~\ref{thm:emptiness} yields the bounds
stated in Theorem~\ref{thm:ceupper}.
\end{proof}

\medskip
\noindent
{\bf Theorem~\ref{thm:unifint2expupper}.}
  It is 2-\ExpTime-complete to decide for a given a TBox \Tmc and
  signature $\Sigma \subseteq \mn{sig}(\Tmc)$, whether there exists a
  uniform $\Sigma$-interpolant of \Tmc.

\medskip
\noindent
\begin{proof}
It remains to prove the lower bound. To this end, 
we reduce deciding conservative extensions to deciding the existence
of uniform interpolants. Assume $\Tmc \subseteq \Tmc'$ are given. 
We may assume that $\Tmc' = \Tmc \cup \{\top \sqsubseteq C\}$ and
that $\Tmc'$ is satisfiable and $\Tmc \not\models \Tmc'$.  Consider the TBox
$$
\Tmc_{0} = \Tmc \cup \{ \neg C \sqsubseteq A, A\sqsubseteq \exists r.A\} \cup 
\{\exists s.A \sqsubseteq A\mid s\in {\sf sig}(\Tmc')\}
$$
where $A$ is a fresh concept name and $r$ a fresh role name.

\medskip
\noindent
Claim. $\Tmc'$ is a conservative extension of $\Tmc$ iff 
there exists a uniform $\Sigma$-interpolant of $\Tmc_{0}$ for $\Sigma={\sf sig}(\Tmc)\cup \{r\}$.

\medskip

Assume first that $\Tmc'$ is a conservative extension of $\Tmc$. We
show that $\Tmc$ is a uniform $\Sigma$-interpolant of $\Tmc_{0}$. 
By Theorem~\ref{bisimuniform}, it is sufficient to
show the following for all $\Imc$: $\Imc \models \Tmc$ iff for all
$d\in \Delta^{\Imc}: (\mathcal{I},d)\models
\exists^{\sim}_{\Sigma}.\Tmc_{0}$.  The direction from right to left
is trivial.  Assume now that $\Imc\models \Tmc$ and fix a $d \in
\Delta^\Imc$. By Theorem~\ref{bisimCE}, $(\mathcal{I},d)\models
\exists^{\sim}_{{\sf sig}(\Tmc)}.\Tmc'$ and thus there is a pointed
model $(\Jmc,e)$ of $\Tmc'$ such that $(\mathcal{I},d)\sim_{{\sf
    sig}(\Tmc)}(\Imc_{d},d)$. As $\top \sqsubseteq C\in \Tmc'$, we
have $C^{\Jmc}= \Delta^{\Jmc}$.  Moreover, since $A$ and $r$ do not
occur in \Tmc and $\Tmc'$, we may assume that
$A^{\Jmc}=\emptyset$. But from $C^{\Jmc}= \Delta^{\Jmc}$
and $A^{\Jmc}=\emptyset$ we obtain $\Jmc\models \Tmc_{0}$ and
so $(\Imc,d) \models \exists^{\sim}_{\Sigma}.\Tmc_{0}$ as required.

Conversely, assume that $\Tmc'$ is not a conservative extension of
$\Tmc$.  Let $C_{0}$ be a ${\sf sig}(\Tmc)$-concept that is
satisfiable w.r.t.~$\mathcal{T}$ but not
w.r.t.~$\mathcal{T}'$. Then
$\mathcal{T}_{0} \models C_{0} \sqsubseteq A$. To show this, let
$\Imc$ be a tree interpretation satisfying $\Tmc_{0}$ with
$\rho^{\Imc} \in C_{0}^{\Imc}$. Let $\Imc_{0}$ be the restriction of
$\Imc$ to all $d\in \Delta^{\Imc}$ that are reachable from
$\rho^{\Imc}$ with paths using roles from ${\sf sig}(\Tmc')$ only.
Then $\Imc_{0}$ is a model of $\Tmc$ and not a model of $\Tmc'$ and so
$C^{\Imc_{0}}\not=\Delta^{\Imc_{0}}$. From $\neg C \sqsubseteq A\in
\Tmc_{0}$ we obtain $A^{\Imc_{0}}\not=\emptyset$.  Thus, from $\exists
s.A\sqsubseteq A\in \Tmc_{0}$ for all $s\in {\sf sig}(\Tmc')$, we
obtain $\rho^{\Imc}\in A^{\Imc_{0}} \subseteq A^\Imc$, as required.

From $A \sqsubseteq \exists r.A\in \Tmc_{0}$, we obtain 
$\Tmc_{0} \models C_{0} \sqsubseteq \exists r.\cdots \exists r.\top$ for arbitrary
long sequences $\exists r.\cdots \exists r$. Intuitively, 
this cannot be axiomatized with a uniform interpolant not using $A$. 
To prove this in a formal way, we apply Theorem~\ref{thm1} and prove
that ($\ast_{m}$) holds for all $m>0$. 

Assume $m>0$ is given.
Let $\Jmc_{0}$ be a tree model satisfying $\Tmc$ with $C_{0}^{\Jmc_{0}}\in \rho^{\Jmc_{0}}$
and in which each $d\in \Delta^{\Imc}$ with the exception of $\rho^{\Jmc_{0}}$ has an
$r$-successor. Let $\Jmc_{i}$ be tree models satisfying $\Tmc'$ with 
$r^{\Jmc_{i}}=\emptyset$, $i>0$. 
We may assume that the $\Delta^{\Jmc_{i}}$ are mutually disjoint.
Define $\Imc_{1} = (\Delta^{\Imc_{1}},\cdot^{\Imc_{1}})$
by setting
\begin{itemize}
\item $\Delta^{\Imc_{1}}= \bigcup_{i\geq 0} \Delta^{\Jmc_{i}}$,
\item $s^{\Imc_{1}}= \bigcup_{i\geq 0} s^{\Jmc_{i}}$ and $B^{\Imc_{1}} = \bigcup_{i\geq 0}B^{\Jmc_{i}}$
for all $s,B\in {\sf sig}(\Tmc')$;
\item $r^{\Imc_{1}}= r^{\Jmc_{0}} \cup \{ (\rho^{\Jmc_{i}},\rho^{\Jmc_{i+1}}) \mid i\geq 0\}$.
\end{itemize}
Note that $\rho^{{\Imc}_{1}}= \rho^{\Jmc_{0}}$. 
Then $(\Imc_{1},\rho^{\Imc_{1}}) \models \exists^{\sim}_{\Sigma}.\Tmc_{0}$ because
the extension of $\Imc_{1}$ defined by setting 
$A^{\Imc_{1}}= \Delta^{\Jmc_{0}}\cup \{\rho^{\Jmc_{i}}\mid i>0\}$
is a model of $\Tmc_{0}$.

Let $\Imc_{2}$ be the restriction of $\Imc_{1}$ to 
$\bigcup_{0\leq i \leq m}\Delta^{\Jmc_{i}}$.
Then $(\Imc_{2},\rho^{\Imc_{2}}) \not\models \exists^{\sim}_{\Sigma}.\Tmc_{0}$:
for any interpretation $\Imc_{2}'$ with $(\Imc_{2},\rho^{\Imc_{2}}) \sim_{\Sigma}
(\Imc_{2}',\rho^{\Imc_{2}'})$ we have $\rho^{\Imc_{2}'} \in C_{0}^{\Imc_{2}'}$. 
If $\Imc_{2}'$ is a model of $\Tmc_{0}$, then $\rho^{\Imc_{2}'} \in A^{\Imc_{2}'}$ and so 
there exist $d_{0},d_{1},\ldots$ such that $d_{0}=\rho^{\Imc_{2}'}$ and 
$(d_{i},d_{i+1})\in r^{\Imc_{2}'}$ for $i\geq 0$. But then there exists such a sequence in 
$\Imc_{2}$ starting at $\rho^{\Imc_{2}}$. As such a sequence does not exist,
we have derived a contradiction. 

On the other hand, for all sons $d$ of $\rho^{\Imc_{2}}$, we have $(\Imc_{2},d)\models
\exists^{\sim}_{\Sigma}.\Tmc_{0}$: for $d=\rho^{\Jmc_{1}}$ this is witnessed by the interpretation
obtained from $\Imc_{2}$ by interpreting $A$ as the empty set. 
For all $d\in \Delta^{\Jmc_{0}}$ this is witnessed by the interpretation obtained
from $\Imc_{2}$ by interpreting $A$ as the whole domain. 

It follows that $\Imc_{1}$ and $\Imc_{2}$ satisfy the condition 
($\ast_{m}$) from Theorem~\ref{thm1}.
\end{proof}

\section{Proofs for Section~\ref{sect:lower}}

\noindent
{\bf Theorem~\ref{thm:internalization}.}
Let $\Tmc= \{\top \sqsubseteq C_\Tmc\}$ and assume that $\Tmc$ has a uniform $\Sigma$-interpolant
Let $R$ denote the set of role names in $\Tmc$,
$m = 2^{2^{|C_\Tmc|+1}} + 2^{|C_\Tmc|} +2$ and let $C$ be a $\Sigma$-concept uniform interpolant of
$\forall R^{\leq m}.C_\Tmc$ w.r.t.~$\Sigma$.
Then $\Tmc'= \{ \top \sqsubseteq C\}$ is a uniform $\Sigma$-interpolant of $\Tmc$.

\medskip
\noindent
\begin{proof}
Recall that $M_{\Tmc}= 2^{2^{|C_\Tmc|}}$. 
By Theorem~\ref{fixm}, $\Tmc_{\Sigma,M_{\Tmc}^{2}+1}$ is a uniform $\Sigma$-interpolant of 
$\Tmc$. We may assume that $\Tmc_{\Sigma,M_{\Tmc}^{2}+1} = 
\{\top \sqsubseteq F\}$ for a $\Sigma$-concept $F$ with ${\sf rd}(F)\leq M_{\Tmc}^{2}+1$. 
We show
$$
\emptyset \models \forall R^{\leq m}.C_\Tmc \sqsubseteq F
$$
We provide a sketch only since the argument is similar to the standard 
reduction of ``global consequence'' to ``local consequence'' in modal logic.
Suppose this is not the case. Let $\Imc$ be a tree interpretation
with $\rho^{\Imc} \in (\forall R^{\leq m}.C_\Tmc)^{\Imc}$ and $\rho^{\Imc} \not\in F^{\Imc}$.
Let $W$ be the set of $d\in \Delta^{\Imc}$ that are of depth $2^{2^{|C_\Tmc|+1}}$. For any path of
length $2^{|C_\Tmc|}+1$ starting at some $d\in W$, there exist at least two points
on that path, say $d_{1}$ and $d_2$, such that
$$
\{ E \in {\sf sub}(C_\Tmc) \mid d_{1} \in E^{\Imc}\} =  \{ E \in {\sf sub}(C_\Tmc) \mid d_{2}\in E^{\Imc}\}.
$$
We remove the subtree $\Imc(d_{2})$
from $\Imc$ and add the pair $(d',d_{1})$ to $r^{\Imc}$ for the unique predecessor $d'$ of 
$d_{2}$ with $(d',d_{2}) \in r^{\Imc}$ for some role $r$. 
This modification is repeated until a (non-tree!) 
interpretation $\Imc'$ is reached in which all points are reachable from $\rho^{\Imc}$ by a 
path of length bounded by $m = 2^{2^{|C_\Tmc|+1}} + 2^{|C_\Tmc|} +2$. Since ${\sf rd}(F) \leq M_{\Tmc}^{2}+1$
and $\Imc$ has not changed for points of depth not exceeding $M_{\Tmc}^{2}+1$,
we still have $\rho^{\Imc} \not\in F^{\Imc'}$. By construction,
$\Imc'$ is a model of $\Tmc = \{\top \sqsubseteq C_\Tmc\}$.
Thus, we have obtained a contradiction to the assumption that $\{\top \sqsubseteq F\}$ is a 
uniform $\Sigma$-interpolant of $\Tmc$. 

From $\emptyset \models \forall R^{\leq m}.C_\Tmc \sqsubseteq F$ we obtain $\emptyset
\models C \sqsubseteq F$ for the $\Sigma$-concept uniform interpolant $C$. 
Thus $\{\top \sqsubseteq C\} \models \top \sqsubseteq F$
and so $\{\top \sqsubseteq C\}$ is a uniform $\Sigma$-interpolant of $\Tmc$.
\end{proof}

\medskip
\noindent
{\bf Theorem~\ref{thm:trplexpsize}.}
There exists a signature $\Sigma$ and a family of TBoxes $(\Tmc_n)_{n
  > 0 }$ such that, for all $n > 0$,
\begin{itemize}


\item[(i)] $|\Tmc_n| \in \Omc(n^2)$ and

\item[(ii)] every uniform $\Sigma$-interpolant $\{ \top \sqsubseteq
  C_\Tmc \}$ for $\Tmc_n$ is of size at least $2^{(2^n \cdot
    2^{2^n})-2}$.
\end{itemize}
To prove Theorem~\ref{thm:trplexpsize} in an economic way, we reuse
some techniques and result from~\cite{KR06}. We first need a bit of
terminology. If $\Tmc$ and $\Tmc'$ are TBoxes and $\Tmc'$ is not a
conservative extension of~\Tmc, then there is a $\Sigma$-concept $C$
such that $C$ is satisfiable relative \Tmc, but not relative to
$\Tmc'$; such a concept $C$ is a \emph{witness concept} for
non-conservativity of the extension of \Tmc with $\Tmc'$.  One main
result of \cite{KR06} is as follows.
\begin{theorem}[\cite{KR06}]
\label{thm:trplexpsizeCE}
There are families of TBoxes $(\Tmc_n)_{n > 0 }$ and $(\Tmc'_n)_{n > 0 }$ 
such that, for all $n > 0$,
\begin{itemize}

\item[(i)] $\Tmc_n \cup \Tmc'_n$ is not a conservative extension of
  $\Tmc_n$,

\item[(i)] $|\Tmc_n| \in \Omc(n^2)$, $|\Tmc'_n| \in \Omc(n^2)$, and

\item[(ii)] every witness concept for non-conservativity of the
  extension of $\Tmc_n$ with $\Tmc'_n$ is of size at least $2^{(2^n
    \cdot 2^{2^n})-1}$.

\end{itemize}
\end{theorem}
To transfer Theorem~\ref{thm:trplexpsizeCE} from witness concepts to
uniform interpolants, we need to introduce some technicalities from
its proof.  For the reminder of this section, fix a signature
$\Sigma=\{A,B,r,s\}$.  Let \Imc be an intepretation and $d \in
\Delta^\Imc$. A \emph{path} starting at $d$ is a sequence
$d_1,\dots,d_k$ with $d_1=d$ and $(d_1,d_{1+1}) \in r^\Imc \cup
s^\Imc$, for $1 \leq i \leq k$.
%
%
In \cite{KR06}, \Imc is called \emph{strongly $n$-violating} iff there
exists an $x \in A^\Imc$ such that the following two properties are
satisfied, where $m=(2^n \cdot 2^{2^n})$:
\begin{itemize}

\item[(P1)] for all paths $x_1,\dots,x_k$ in \Imc with $k \leq m$
  starting at $x$, the $X$ values of $x_1,\dots,x_k$ describe the
  first $k$ bits of a $2^n$-bit counter counting from $0$ to
  $2^{2^n}-1$.

\item[(P2)] 
  there exist elements $x_w \in \Delta^\Imc$, for all
  $w \in \{r,s\}^*$ of length at most $m-1$, such that the following
  are true:
  \begin{itemize}

  \item[($a$)] $x_\varepsilon = x$;

  \item[($b$)] $(x_w,x_{w'}) \in r^\Imc$ if $w'=w \cdot r$, and $(x_w,x_{w'})
    \in s^\Imc$ if $w'=w \cdot s$;

  \item[($c$)] $x_w \notin B^\Imc$ if $w$ is of length $m-1$.

  \end{itemize}

\end{itemize}
Define a 
$\Sigma$-TBox
$$
\begin{array}{r@{}l}
  \Tmc_n^- = \{ & \top \sqsubseteq \forall r . \neg A \sqcap \forall s . \neg A \\[2mm]
  &  A \sqsubseteq \neg X \sqcap \bigsqcap_{i<2^n} \forall (r \cup s)^i. \neg X \}
\end{array}
$$
The following result of \cite{KR06} underlies the proof of Theorem~\ref{thm:trplexpsizeCE}.
\begin{lemma}[\cite{KR06}]
\label{thm:modelstuff}
There exist families of TBoxes $(\Tmc_n)_{n > 0 }$ and $(\Tmc'_n)_{n >
  0 }$ such that, for all $n > 0$,
\begin{itemize}


\item[(i)] $|\Tmc_n| \in \Omc(n^2)$, $|\Tmc'_n| \in \Omc(n^2)$;

\item[(ii)] a model of $\Tmc_n$ that is strongly $n$-violating cannot be extended
  to a model of $\Tmc'_n$;

\item[(iii)] a tree model of $\Tmc_n$ that is not strongly $n$-violating can be extended
  to a model of~$\Tmc'_n$;

\item[(iv)] every model of $\Tmc_n^-$ can be extended to a model of $\Tmc_n$.

\end{itemize}
\end{lemma}
The TBoxes $\Tmc_n$ and $\Tmc'_n$ from Lemma~\ref{thm:trplexpsizeCE}
are formulated in extensions of the signature~$\Sigma$, more precisely
we have $\Sigma=\mn{sig}(\Tmc_n) \cap \mn{sig}(\Tmc'_n)$. Thus, the phrase
`extended to a model of' refers to interpreting those symbols that do not
occur in the original TBox.

\smallskip

To estabish Theorem~\ref{thm:trplexpsize}, we consider the
uniform $\Sigma$-interpolants of the TBoxes $\Tmc_n \cup \Tmc'_n$. Let
$$
  \Tmc_{\Sigma,n} = \Tmc^-_n \cup \{ A \sqsubseteq \neg K_1 \sqcup \neg K^{(m)}_2 \}
$$
\psfull
\begin{figure}[t!]
  \begin{boxedminipage}[t]{\columnwidth}
\vspace*{-2mm}
  \begin{center}
\begin{eqnarray}
  K_1 & = & \bigsqcap_{i<m, \mn{bit}_i(2^{2^n})=1} \forall \{r,s\}^i . X \; \sqcap \\[1mm]
  &&
  \bigsqcap_{i<m, \mn{bit}_i(2^{2^n})=0} \forall \{r,s\}^i . \neg X \\
  K^{(0)}_2 & = & \neg B \\
  K^{(i+1)}_2 & = & \exists r . K^{(i)}_2 \sqcap \exists s . K^{(i)}_2
\end{eqnarray}
  \end{center}
  \end{boxedminipage}
  \caption{Definition of the concepts $K_1$ and $K^{(m)}_2$.}
  \label{fig:tbox2}
\end{figure}
\psdraft
where the concepts $K_1$ and $K^{(m)}_2$ are shown in
Figure~\ref{fig:tbox2} and $m = 2^n \cdot 2^{2^n}$. In the figure, we
use $\mn{bit}_i(m)$ to denote the $i$-th bit of the string obtained by
concatenating all values of a binary counter that counts up to $m$
(lowest bit first, and with every counter value padded to $\log(m)$
bits using trailing zeros). Note that $\forall \{r \cup s\}^i . C$
is an abbreviation for $\bigsqcap_{r_1 \cdots r_i \in \{r,s\}^i}
\forall r_1 . \cdots . \forall r_i . C$. We show that $\Tmc_{\Sigma,n}$
is a uniform $\Sigma$-interpolant of $\Tmc_n \cup \Tmc'_n$, and that
it is essentially of minimal size. It is not hard to see that the
models of $\Tmc_{\Sigma_n}$ are precisely those interpretations that are
not strongly $n$-violating.
\begin{lemma}
  For all $n \geq 0$,
  \begin{enumerate}

  \item $\Tmc_{\Sigma,n}$ is a uniform $\Sigma$-interpolant of $\Tmc_n
    \cup \Tmc'_n$;

  \item every uniform $\Sigma$-interpolant $\{ \top \sqsubseteq C_\Tmc
    \}$ for $\Tmc_n \cup \Tmc'_n$ is of size at least $2^{(2^n \cdot
      2^{2^n})-1}$.

  \end{enumerate}
\end{lemma}
\begin{proof}
  For Point~1, let $C \sqsubseteq D$ be a $\Sigma$-inclusion and
  assume first that $\Tmc_{\Sigma,n} \not\models C \sqsubseteq D$,
  i.e., there is a model $\Imc$ of $\Tmc_{\Sigma,n}$ and a $d \in (C
  \sqcap \neg D)^\Imc$. Since $\Tmc^-_n \subseteq \Tmc_{\Sigma,n}$ and
  by Point~(iv) of Lemma~\ref{thm:modelstuff}, $\Imc$ can be extended
  to a model $\Imc'$ of $\Tmc_n$. Let \Jmc be the unravelling of
  $\Imc'$ into a tree with root $d$. Obviously, \Jmc is still a model
  of $\Tmc_n$ and $d \in (C \sqcap \neg D)^\Jmc$.  Since $\Imc \models
  \Tmc_{\Sigma,n}$, \Imc is not strongly $n$-violating, and thus the
  same holds for \Jmc. By Point~(iii) of Lemma~\ref{thm:modelstuff},
  $\Jmc$ can be extended to a model $\Jmc'$ of $\Tmc_n$, and we have
  $d \in (C \sqcap \neg D)^{\Jmc'}$, thus $\Tmc_n \cup \Tmc'_n \not
  \models C \sqsubseteq D$. 

  Conversely, let $\Tmc_n \cup \Tmc'_n \not \models C \sqsubseteq D$.
  Then there is a model $\Imc$ of $\Tmc_n \cup \Tmc'_n$ and a $d \in
  (C \sqcap \neg D)^\Imc$. By (ii), \Imc is not strongly
  $n$-violating, thus it is a model of $\Tmc_{\Sigma,n}$ and we get
  $\Tmc_{\Sigma,n} \not\models C \sqsubseteq D$.

  For Point~2, assume that there is a uniform $\Sigma$-interpolant $\{
  \top \sqsubseteq C_\Tmc \}$ for $\Tmc_n \cup \Tmc'_n$ that is of
  size strictly smaller than $2^{(2^n \cdot 2^{2^n})-2}$. Then the
  size of $\neg C_\Tmc$ is strictly smaller than $2^{(2^n \cdot
    2^{2^n})-1}$. By Theorem~\ref{thm:trplexpsizeCE}, to obtain a
  contradiction it thus suffices to show that $\neg C_\Tmc$ is a
  witness concept for non-conservativity of the extension of $\Tmc_n$
  with $\Tmc'_n$. First, since $\Tmc_n \cup \Tmc'_n \models \top
  \sqsubseteq C_\Tmc$, $\neg C_\Tmc$ is unsatisfiable relative to
  $\Tmc_n \cup \Tmc'_n$. And second, there clearly is a model \Imc of
  $\Tmc^-_n$ that is not a model of $\Tmc_{\Sigma,n}$. Since $\{ \top
  \sqsubseteq C_\Tmc \}$ and $\Tmc_{\Sigma,n}$ are both uniform
  $\Sigma$-interpolants of $\Tmc_n \cup \Tmc'_n$, they are equivalent
  and thus there is a $d \in \neg C_\Tmc^\Imc$. By Point~(iv) of
  Lemma~\ref{thm:modelstuff}, $\neg C_\Tmc$ is satisfiable relative to
  $\Tmc_n$.
\end{proof}


\end{document}